\documentclass[oneside,a4paper,12pt]{article}
\usepackage[utf8]{inputenc} %solo en ingles
\usepackage{lmodern}
\usepackage[T1]{fontenc}
\usepackage[svgnames,x11names,dvipsnames]{xcolor} % for colours
\usepackage{amsfonts,amsmath,latexsym,amssymb,amsthm,epigraph, bbm, mathtools, upgreek}  %math symbols
\usepackage{xparse, stmaryrd}
\usepackage{tikz-cd, tikz} % commutative diagrams
\usepackage{enumerate} %change enumeration items
\usepackage[]{hyperref} 
\usepackage{cleveref} % intelligent referencing
\usepackage{indentfirst} %sangrado en todos los parrafos
\usepackage[text={16cm,23.5cm}]{geometry} %tama?o cuadro texto
\usepackage{breqn}
\usepackage{qtree}
\usepackage{faktor}
\usepackage{pifont}
\usepackage{multicol}
\usepackage{caption}
\usepackage{subcaption}
\usepackage{float}
\usepackage{array}

\usepackage{diagbox}
\newcommand{\N}{{\mathbb{N}}} % natural numbers
\newcommand{\C}{{\mathbb{C}}} % complex numbers
\newcommand{\F}{{\mathbb{F}}}
\DeclareMathOperator{\dis}{d}
\DeclareMathOperator{\Dis}{D}
\DeclareMathOperator{\supp}{supp}
\DeclareMathOperator{\ev}{ev}
\DeclareMathOperator{\w}{w}
\DeclareMathOperator{\Span}{span}

\NewDocumentCommand{\var}{g}{% Ejemplo: \var = Var_k, \var[K] = Var_K
	\ensuremath{\operatorname{Var}_{\IfNoValueTF{#1}{k}{#1}}}%
}

\usepackage{mathtools}
\RequirePackage{etoolbox}
\DeclarePairedDelimiterX\conj[1]{\lbrace}{\rbrace}{#1} % for sets
\DeclarePairedDelimiterX\paren[1]{\lparen}{\rparen}{#1} % for parenthesis
\DeclarePairedDelimiterX\braces[1]{\lbrace}{\rbrace}{\ifblank{#1}{\:\cdot\:}{#1}} % corchetes
\DeclarePairedDelimiterX\scalar[1]{\langle}{\rangle}{\ifblank{#1}{\:\cdot\:}{#1}} % scalar product from duality of lattices
\DeclarePairedDelimiterX\card[1]{\lvert}{\rvert}{#1} % cardinal of a set
\DeclarePairedDelimiterX\gen[1]{\langle}{\rangle}{\ifblank{#1}{\:\cdot\:}{#1}} % subspace generated by
\DeclarePairedDelimiterX\abs[1]{\mid}{\mid}{#1}% absolute value

\DeclarePairedDelimiter\floor{\lfloor}{\rfloor}

\theoremstyle{plain} %estilo en cursiva
\newtheorem{theorem}{Theorem}[section]
\newtheorem{lemma}[theorem]{Lemma}
\newtheorem{corollary}[theorem]{Corollary}
\newtheorem{proposition}[theorem]{Proposition}

\theoremstyle{definition} %estilo plano
\newtheorem{definition}[theorem]{Definition}
\newtheorem{example}[theorem]{Example}%
\newtheorem{remark}[theorem]{Remark}

\title{MDS, Hermitian Almost MDS, and  Gilbert-Varshamov  Quantum Codes from Generalized Monomial-Cartesian Codes}
\author{Beatriz Barbero-Lucas\footnote{School of Mathematics and Statistics, University College Dublin, Ireland}, Fernando Hernando\footnote{Instituto Universitario de Matem\'aticas y Aplicaciones de Castell\'on and Departamento de Matem\'aticas, Universitat Jaume I, Campus de Riu Sec, 12071 Castell\'o, Spain},\\ Helena Mart\'in-Cruz\footnote{Instituto Universitario de Matem\'aticas y Aplicaciones de Castell\'on and Departamento de Matem\'aticas, Universitat Jaume I, Campus de Riu Sec, 12071 Castell\'o, Spain},  Gary McGuire\footnote{School of Mathematics and Statistics, University College Dublin, Ireland}}
\date{\today}
\hypersetup{ % rellenar titulo, autor y colores de los enlaces
	pdftitle={MDS, Hermitian Almost MDS, and  Gilbert-Varshamov  Quantum Codes from Generalized Monomial-Cartesian Codes},
	pdfauthor={Beatriz Barbero-Lucas, Fernando Hernando,\\ Helena Mart\'in-Cruz,  Gary McGuire},
	bookmarksnumbered=true,     
	bookmarksopen=true,         
	bookmarksopenlevel=1,       
	colorlinks=true,
	linkcolor=blue,
	citecolor=Green,
	urlcolor=blue,
	pdfstartview=Fit,           
	pdfpagemode=UseOutlines,    
}

\begin{document}

\maketitle

\begin{abstract}
    We construct new stabilizer quantum error-correcting codes from generalized monomial-Cartesian codes. 
    Our construction uses an explicitly defined  twist vector, and we present formulas
    for the minimum distance and dimension. 
    Generalized monomial-Cartesian codes arise from polynomials in $m$ variables.
    When $m=1$ our codes are MDS, and when $m=2$ 
    and our lower bound for the minimum distance is $3$ the codes are at least Hermitian Almost MDS. 
    For an infinite family of parameters when $m=2$
    we prove that our codes beat the Gilbert-Varshamov bound. 
    We also present many examples of our codes that are better
    than any known code in the literature.
\end{abstract}

\section{Introduction}

Certain classically intractable problems can become feasible when approached with the computational power of quantum computers. This was demonstrated through Shor's algorithm which solves in polynomial time the prime factorization problem and discrete logarithm problem on quantum computers \cite{Shor2}. Due to this fact, researchers and companies are actively engaged in constructing quantum computers with many qubits \cite{Cas2017,Brooks23}. Quantum computer implementations have higher error rates than classical computers, making reliability a challenge. However, despite quantum information being unclonable \cite{26RBC,8AS}, it was shown that quantum error correction techniques can be used \cite{23RBC,95kkk}.  Over the last twenty-five years, error-correction has proved to be one of the main obstacles to scaling up quantum computing and quantum information processing.

There is an extensive study of quantum error-correcting codes, see for example the papers \cite{Gottesman,St96,18kkk,Calderbank,45kkk,ABKL2000I,ABKL2000II} for the binary case and \cite{BE,71kkk,ALT2001,AK,Feng2002,GBR2004,Ketkar,lag1,lag2,GGHR2017,Song,CaoCui,GHMR2023} for the general case.
Many of the known quantum error-correcting codes are stabilizer codes. Let $\C$ be the complex field, let $q$ be a prime power and let $n$ be a positive integer. A stabilizer code $Q\neq \{0\}$ is the common eigenspace of an abelian subgroup of the error group $G_n$ generated by a nice error basis on the space $\C^{q^n}$. The code $Q$ has minimum distance $d$ whenever all errors in $G_n$ with weight less than $d$ can be detected, or have no effect on $Q$, but some errors of weight $d$ cannot be detected. A code as above has parameters $[[n,k,d]]_q$ when it is a $q^k$-dimensional subspace of $\C^{q^n}$ and has minimum distance $d$ (see, for instance, \cite{Calderbank,Ketkar}). 
Stabilizer quantum error-correcting codes have been studied by many authors because they can be constructed  from classical additive codes  in $\mathbb{F}_q^{2n}$ which are self-orthogonal with respect to a trace symplectic form. 
In particular, stabilizer codes can be obtained from suitable Hermitian self-orthogonal classical linear codes (see \cite{Ketkar} or \cite{Calderbank,BE,AK} for details). 
We will utilize this construction.

Many constructions of classical codes
start with a quotient polynomial ring of the form
$\F_q[X_1,\ldots ,X_m]/I$ where $I$ is an ideal.
    Affine variety codes were introduced by Fitzgerald and Lax in \cite{fit},
    with a  general ideal $I$.
    Our codes $C_{\boldsymbol{v},\Delta,Z}$ (defined in the next section) are a type of generalized affine variety code, so we could use this name. However, since the codes we define are generalized monomial-Cartesian codes, introduced in \cite{LMS2020}, and although the definition is slightly different, we are going to call our codes $C_{\boldsymbol{v},\Delta,Z}$
     \textit{generalized monomial-Cartesian codes}.

    Monomial-Cartesian codes (MCCs) are a class of evaluation codes obtained as the image of maps
$$
\ev_S\colon V_\Delta \subset \mathbb{F}_q[X_1,\dots,X_m]/I \longrightarrow  \mathbb{F}_q^{n} \textrm{, } \quad \ev_S(f)=\left(f(\boldsymbol{\beta}_1),\dots,f(\boldsymbol{\beta}_n)\right),
$$
where $m$ is a positive integer larger than $1$, $S=S_1\times\cdots\times S_m=\{\boldsymbol{\beta}_1,\dots,\boldsymbol{\beta}_n\}$ is a Cartesian-product subset of $\mathbb{F}_q^m$, $I$ is the vanishing ideal at $S$ of $\mathbb{F}_q[X_1,\dots,X_m]$, and $V_\Delta$ is an $\mathbb{F}_q$-linear space generated by classes of monomials. MCCs were introduced in \cite{LMS2020} with only algebraic tools, see also  \cite{LSV2021}. These codes have several  different applications in the literature, such as quantum codes, LRCs with availability, polar codes and $(r,\delta)$-LRCs \cite{LMS2020,CLMS2021,GHMC2023}. 

Generalized monomial-Cartesian codes arise when changing the evaluation map $\ev_S$ to twist each coordinate of $\ev_S(f)$ by nonzero elements of $\F_q$. In this article we will use generalized MCCs, where the set $S_1$ is a certain fixed set, and we will use the same name for this construction, see \Cref{gmcc}.
We will use generalized monomial-Cartesian codes to 
construct Hermitian self-orthogonal classical linear codes,
and thereby construct stabilizer quantum codes.
We present some evidence comparing our codes to codes in \cite{ternary,GHR2015MPC,MatrixProductCodes2018,Kolotolu2019QUANTUMCW,CaoCui,HermitianDualConstacyclic2021,Constacyclic2022} which shows that they are very good 
quantum codes, and sometimes optimal.

Quantum MDS codes are those achieving the quantum Singleton bound, there are many papers on this type of codes (some recent papers are \cite{Fang,Ball,Liu-LiuX}). 
The MDS conjecture limits the length of a $q$-ary quantum MDS code to be at most $q^2+2$ \cite{Ketkar}. Thus, another goal is to obtain longer $q$-ary codes with good parameters. With our  construction, we achieve this.

The paper is laid out as follows.
After the preliminaries in \Cref{preliminaries}, we present our construction in \Cref{gen}.
Previous works using a twist vector have proved the existence of a twist
vector with the required properties, whereas a feature of our construction
is that we define the twist vector explicitly, see \eqref{twist1} in \Cref{gen}.
We present a general construction first (\Cref{te:SelfOrthogonal2}), and then a more specific construction that 
allows us to control the minimum distance (\Cref{main2}).
In \Cref{secMDS} we will show that our construction with $m=1$ gives MDS codes.
We also prove that when  $m=2$ 
    and our lower bound for the minimum distance is 3 the codes are at least Hermitian Almost MDS. 
    \Cref{secGV} contains a proof that for an infinite family of parameters when $m=2$,
     our codes beat the Gilbert-Varshamov bound. 
     Finally, in \Cref{secex} we present some examples with small parameters that 
     beat the best known codes in the literature. 

\section{Preliminaries}\label{preliminaries}

In this paper, we will assume $q$ is odd, although in this section the definitions hold for any $q$.
Let us denote by $\N$ the set of positive integers and by $\N_0$ the set of nonnegative integers. For any two vectors $\boldsymbol{a}=(a_0,\dots,a_{n-1})$, $\boldsymbol{b}=(b_0,\dots,b_{n-1}) \in \F_{q^2}^n$, their Hermitian inner product is defined as
$$\boldsymbol{a} \cdot_h \boldsymbol{b}=\sum_{i=0}^{n-1} a_ib_i^q,$$
their Euclidean inner product is defined as
$$\boldsymbol{a} \cdot_e \boldsymbol{b}=\sum_{i=0}^{n-1} a_ib_i,$$
and their * product is defined as 
$$(a_0,\dots,a_{n-1})*(b_0,\dots,b_{n-1})=(a_0\cdot b_0,\dots,a_{n-1}\cdot b_{n-1}).$$
Let the symbol $\perp_h$ (respectively, $\perp_e$) mean dual with respect to Hermitian (respectively, Euclidean) inner product. 
For a vector subspace (or code) $C$ of $\F_{q^2}^n$ we let $C^{\perp_h}$ (respectively, $C^{\perp_e}$) denote the orthogonal vector subspace (the dual code) with respect to the Hermitian (respectively, Euclidean) inner product. We denote by $\dis(C)$ the minimum distance of $C$. Let $s$ be a nonnegative integer and $\boldsymbol{c}=(c_0,\dots,c_{n-1})\in C$ be a codeword. We denote $\boldsymbol{c}^{s}=(c_0^s,\dots,c_{n-1}^s)$ and
$$C^s:=\{\boldsymbol{c}^s \mid \boldsymbol{c}\in C\}\subseteq \F_{q^2}^n.$$
Let us denote by $\w(\boldsymbol{c})$ the Hamming weight of $\boldsymbol{c}$. We say that two codes are isometric if there exists a bijective mapping between them that preserves Hamming weights.

\begin{theorem}[\cite{Aly, Ketkar}]\label{th: quantherm} 
    Let $C$ be a linear $[n,k,d]$ error-correcting code over the field $\F_{q^2}$ such that $C\subseteq C^{\perp_h}$. Then, there exists an $[[n,n-2k,\geq d^{\perp_h}]]_q$ stabilizer quantum code, where $d^{\perp_h}$ stands for the minimum distance of $C^{\perp_h}$.
\end{theorem}

The idea in this paper is to construct codes that satisfy the hypotheses of \Cref{th: quantherm}. In order to do so, we fix a finite field $\F_{q^2}$.  Let $\F_{q^2}[X_1,\dots,X_m]$ be the polynomial ring in $m\geq 1$ variables over $\F_{q^2}$. For each element $\boldsymbol{e}=(e_1,\dots,e_m)\in\N_0^m$, we write $X^{\boldsymbol{e}}$ for $X_1^{e_1}X_2^{e_2}\cdots X_m^{e_m}$. 
We will refer to $\boldsymbol{e}$ as an exponent and use the lexicographic order in $\N_0^m$ for the exponents. That is, given $\boldsymbol{e}$, $\boldsymbol{e'}\in \N_0^m$, we say $\boldsymbol{e}<\boldsymbol{e'}$ if and only if $e_1<e'_1$ or there exists $j\in\{2,\dots, m\}$ such that $e_1 =e'_1, \dots , e_{j-1} = e'_{j-1}$ and $e_j<e'_j$. Any order can be used.

Let $\lambda \in \N$ such that $\lambda \mid q-1$. Let $A_1$ be the set of roots of the polynomial $X_1^{\lambda (q+1)}-1$, which lie in $\F_{q^2}$. We also consider arbitrary subsets $A_j \subseteq \F_{q^2}^*$ for $j=2, \dots, m$ which have cardinality greater than or equal to 2. Let $a_j \coloneqq \# A_j$ for $j=1, \ldots, m$,
so that $a_1=\lambda (q+1)$. Let
\[
    Z \coloneqq A_1  \times \cdots \times A_m,
\]
which has cardinality
$$n:=\prod_{j=1}^m a_j.$$
Let 
$$Q_j(X_j)=\prod_{\beta\in A_j}(X_j - \beta)$$ 
be the monic polynomial in one variable whose roots are the elements of $A_j$, then $\deg(Q_j)=  a_j$ for $j=1,\dots,m$. Let $I$ be the ideal of $\F_{q^2}[X_1,\dots,X_m]$ generated by the polynomials 
$Q_1(X_1)=X_1^{\lambda (q+1)}-1$ and $Q_j(X_j)$ for $j=2, \ldots , m$. Let
$$R:=\faktor{\F_{q^2}[X_1,\dots,X_m]}{I}$$
and let
$$E:=\{0,1,\dots,a_1-1\}\times \cdots \times \{0,1,\dots,a_m-1\}.$$ 
Given $f\in R$, in this paper $f$ is going to denote both the equivalence class in $R$ and the unique polynomial representing $f$ in $\F_{q^2}[X_1,\dots,X_m]$ with degree in $X_j$ less than $a_j$, $1\leq j \leq m$. 
Thus, one can write any $f\in R$ uniquely as
$$f(X_1,\dots,X_m)=\sum_{(e_1,\dots,e_m)\in E} f_{e_1,\dots,e_m}X_1^{e_1}\cdots X_m^{e_m},$$
with $f_{e_1,\dots,e_m}\in\F_{q^2}$. 
Let us denote $\supp(f)=\{(e_1,\dots,e_m)\in E \mid f_{e_1,\dots,e_m}\neq 0\}$. For each nonempty subset $\Delta \subseteq E$, define $V_\Delta:=\{f\in R \mid \supp(f)\subseteq \Delta\}$. Then, $V_\Delta$ is the $\mathbb{F}_{q^2}$-vector space $\Span\{ X^{\boldsymbol{e}} \mid \boldsymbol{e}\in \Delta\}$.

For any positive integer $t$, we denote by $\zeta_t$ a primitive $t$-th root of unity. Since $A_j$ has $a_j$ elements, we choose a bijection between $A_j$ and the set $\{0,1, \dots , a_j-1\}$ and this is going to give us an ordering of $A_j$, $j=2,\dots,m$. 
Let us represent by $\xi_{(j,s)}$ the elements of each set $A_j$, where the subindex  $s\in \{0,1, \dots , a_j-1\}$ is given by the  ordering.
For
$\boldsymbol{\alpha} = (\alpha_1, \dots, \alpha_m) \in E$ we define
 $\boldsymbol{P_{\alpha}} \in Z$ by
\[
    \boldsymbol{P_{\alpha}}\coloneqq (\zeta_{\lambda(q+1)}^{\alpha_1}, \xi_{(1,\alpha_2)}, \dots \xi_{(m,\alpha_m)}),
\]
where $\alpha_1$ indicates the exponent of $\zeta_{\lambda(q+1)}$ and $\alpha_j \in \{0,1, \ldots ,a_j-1\}$ gives the position of the element $\xi_{(j,\alpha_j)}\in A_j$
in the ordering of $A_j$, $j=2,\dots,m$.
Every element of $Z$ has the form $\boldsymbol{P_{\alpha}}$ for some $\boldsymbol{\alpha}\in E$.
This sets up a bijection between $Z$ and $E$.

We order the set $Z$ using the (lexicographic) order in $\N_0^m$ restricted to $E$. 
That is, given  $\boldsymbol{P_\alpha}$, $\boldsymbol{P_{\alpha'}} \in Z$, then $\boldsymbol{P_\alpha} < \boldsymbol{P_{\alpha'}}$ if and only if $\boldsymbol{\alpha}<\boldsymbol{\alpha'}$. 
Then, we can rename the points in $Z$ as 
$$\boldsymbol{P_0}:= \boldsymbol{P}_{(0,\dots,0)}, \boldsymbol{P_1}:= \boldsymbol{P}_{(0,\dots,0,1)}, \dots , \boldsymbol{P_{n-1}}:= \boldsymbol{P}_{(a_1-1,a_2-1,\dots,a_m-1)}.$$ 

Let $\boldsymbol{v}=(v_0,\dots,v_{n-1})\in(\F^*_{q^2})^{n}$,  we will refer 
to this vector as the \textit{twist vector}. 
We index the coordinates of $\boldsymbol{v}$ by the elements of $E$,
and we  order the coordinates of $\boldsymbol{v}$ 
in the same way as we ordered the elements of $Z$.
That is,
$$v_0:=v_{(0,\dots,0)}, v_1:= v_{(0,\dots,0,1)}, \dots , v_{n-1}:= v_{(a_1-1,a_2-1,\dots,a_m-1)}.$$ 
The linear evaluation map in $Z$:
$$
\ev_{\boldsymbol{v},Z}\colon R \longrightarrow \F_{q^2}^{n} \textrm{, } \quad 
\ev_{\boldsymbol{v},Z}(f)=\left(v_0 f(\boldsymbol{P_0}),\dots,v_{n-1}f(\boldsymbol{P_{n-1}})\right)
$$
is injective by the definition  of $R$.
It provides the following class of evaluation codes.

\begin{definition}\label{gmcc}
The \textit{generalized monomial-Cartesian code} (\textit{GMCC}) $C_{\boldsymbol{v},\Delta,Z}$ is the image of $V_\Delta$ via the evaluation map $\ev_{\boldsymbol{v},Z}$, that is
$$C_{\boldsymbol{v},\Delta,Z}:=\ev_{\boldsymbol{v},Z}(V_\Delta)=\Span\{\ev_{\boldsymbol{v},Z}(X^{\boldsymbol{e}}) \mid \boldsymbol{e}\in \Delta\} \subseteq \mathbb{F}_{q^2}^{n}.$$
\end{definition}

Since the order of the set $Z$ will be fixed for the rest of the article, we will use the notation $\ev_{\boldsymbol{v}} \coloneqq \ev_{\boldsymbol{v},Z}$ and $C_{\boldsymbol{v},\Delta}\coloneqq C_{\boldsymbol{v},\Delta, Z}$.

\begin{remark}
Evaluation maps of our codes are defined on subsets
of coordinate rings of certain affine varieties, but these codes can also be introduced with algebraic tools. Monomial-Cartesian codes were introduced in \cite{LMS2020} using only algebraic tools.
When the set $A_1\subseteq \mathbb{F}_{q^2}$ is arbitrary, GMMCs extend monomial-Cartesian codes. This should be the accurate definition, but for our purposes in this paper we use this particular set $A_1$, namely the $\lambda (q+1)$-th roots of unity.
\end{remark}

Here is a standard fact, that the dual of a GMCC is another GMCC.

\begin{lemma}\label{prop:ClosedByDuality}
    The dual code $(C_{\boldsymbol{v},\Delta})^{\perp_h}$ is a GMCC $C_{\boldsymbol{w},\Delta}$ for some twist vector $\boldsymbol{w}$. 
\end{lemma}

\begin{proof}
Assume that $\boldsymbol{v}=(v_0,\dots,v_{n-1})$. Consider the two codewords $\boldsymbol{c}=(c_0,\ldots,c_{n-1})\in C_{\boldsymbol{1},\Delta}$ and $\boldsymbol{a}=(a_0,\ldots,a_{n-1})\in (C_{\boldsymbol{1},\Delta})^{\perp_h}$. Then, the following equation holds:
\begin{equation}\label{eq:condition1}
c_0a_0^q+\cdots+c_{n-1}a_{n-1}^q=0.
\end{equation}
Let us see that $(C_{\boldsymbol{v},\Delta})^{\perp_h}=C_{\boldsymbol{w},\Delta}$ for some $\boldsymbol{w}=(w_0,\dots,w_{n-1})$. To do so, let us find a particular $\boldsymbol{w}$ satisfying this. 
We know  that $\boldsymbol{v}*\boldsymbol{c}=(v_0c_0,\ldots,v_{n-1}c_{n-1})\in C_{\boldsymbol{v},\Delta}$ whenever $\boldsymbol{c}=(c_0,\ldots,c_{n-1})\in C_{\boldsymbol{1},\Delta}$. Moreover, we claim that $\boldsymbol{w}*\boldsymbol{a}=(w_0a_0,\ldots,w_{n-1}a_{n-1})\in (C_{\boldsymbol{v},\Delta})^{\perp_h}$.
To that, it is necessary that 
$$
v_0w_0^qc_0 a_0^q+\cdots+v_{n-1}w_{n-1}^q c_{n-1} a_{n-1}^q=0.
$$
Observe that if $w_i=\frac{1}{v_i^q}$ for all $i=0,\ldots,n-1$, then \eqref{eq:condition1} holds. With this, $$\left(\frac{1}{v_0^q}a_0,\ldots,\frac{1}{v_{n-1}^q}a_{n-1}\right)\in(C_{\boldsymbol{v},\Delta})^{\perp_h}$$ and from the fact that
$$C_{\boldsymbol{1},\Delta} \to C_{\boldsymbol{v},\Delta}, \quad \boldsymbol{c}\mapsto \boldsymbol{v}*\boldsymbol{c}$$
and
$$(C_{\boldsymbol{1},\Delta})^{\perp_h} \to (C_{\boldsymbol{v},\Delta})^{\perp_h}, \quad \boldsymbol{a}\mapsto \boldsymbol{w}*\boldsymbol{a}$$
are bijective mappings, then
$(C_{\boldsymbol{v},\Delta})^{\perp_h}=C_{\boldsymbol{w},\Delta}$.
\end{proof}

The length and the dimension of a GMCC are $n$ and $\#\Delta$, respectively. A bound for the minimum distance is provided in the forthcoming \Cref{cdist}.

\begin{lemma}\label{le:isometricPrimal}
The GMCCs 
$C_{\boldsymbol{1},\Delta}$ and $C_{\boldsymbol{v},\Delta}$ are isometric.
\end{lemma}
\begin{proof}
For any codeword $\boldsymbol{c}=(c_0,\ldots,c_{n-1})\in C_{\boldsymbol{1},\Delta}$, its twisted analogue codeword $\boldsymbol{v}*\boldsymbol{c}=(v_0c_0,\ldots,v_{n-1}c_{n-1})\in C_{\boldsymbol{v},\Delta}$ under the bijective mapping  $C_{\boldsymbol{1},\Delta} \to C_{\boldsymbol{v},\Delta}$, $\boldsymbol{c}\mapsto \boldsymbol{v}*\boldsymbol{c}$ has the same Hamming weight, this is because $v_i\ne 0$ for all $i=1,\dots,n$.
\end{proof}

Affine variety codes admit a bound on the minimum distance, known as the footprint bound \cite{OH}. Monomial-Cartesian codes $C_{\boldsymbol{1},\Delta}$ in the sense of our \Cref{gmcc} (the evaluation map is defined over the coordinate ring of some affine variety) are affine variety codes. This fact and \Cref{le:isometricPrimal} proves next lemma, stating that this bound is also valid for GMCCs. For every exponent $\boldsymbol{e}\in E$, we define
$$\Dis(\boldsymbol{e}):=\prod_{j=1}^m (a_j-e_j).$$

\begin{lemma}\label{footprint}
Let $C_{\boldsymbol{v},\Delta}$ be a GMCC and let $\boldsymbol{c}=\ev_{\boldsymbol{v}}(f)\in C_{\boldsymbol{v},\Delta}$ be a codeword, $f\in R$. Fix a monomial ordering on $(\N_0)^m$ and let $X^{\boldsymbol{e}}$ be the leading monomial of $f$. Then, $\w(\boldsymbol{c})\geq \Dis(\boldsymbol{e})$.
\end{lemma}

\begin{corollary}\label{cdist}
Let $C_{\boldsymbol{v},\Delta}$ be a GMCC and let $d$ be its minimum distance. Define $d_0=d_0\left(C_{\boldsymbol{v},\Delta}\right):=\min\{\Dis(\boldsymbol{e}) \mid \boldsymbol{e}\in \Delta\}$. Then, $d\geq d_0$.
\end{corollary}

\begin{remark}
Affine variety codes were introduced  in \cite{fit} for any ideal  $I$.  A classical result coming from the theory of Gröbner  basis  \cite{CoxLittle} implies that  $d\geq d_0$, where $d$ stands for the minimum distance of an affine variety code and $d_0$ is the cited footprint bound \cite{OH}. Independently, inspired by the algebraic geometric codes \cite{H-VL-P} the so called Feng-Rao bound for the minimum distance of the dual code is derived \cite{FengRao}. It is known that every linear code is an algebraic geometric code. A similar bound (Andersen-Geil) was also given for an algebraic geometric code \cite{ANDERSEN200892}. It turns out that for monomial-Cartesian codes the footprint bound applied to the dual code and the Feng-Rao bound coincide \cite{GGHR}. Although the footprint bound is more natural for the primal code, and the Feng-Rao bound is more natural for the dual code, we will always refer to them as $d_0$.

\end{remark}

For $m=2$, we can use a grid to represent the set $E$ so that an exponent $\boldsymbol{e}=(e_1,e_2)\in E$ correspond to coordinates $(e_1,e_2)$ in the grid, and it is labelled with $\Dis(\boldsymbol{e})$. In addition, exponents in a set $\Delta\subseteq E$ are coloured in blue. 
\begin{example}
\Cref{fig:grid} shows the case when $a_1=8$, $a_2=6$ and $\Delta=\left(\{0,1,2\}\times\{0,1\} \right) \cup \{(0,2),(1,2)\}$. In this example, the lower bound for the minimum distance of the code $C_{\boldsymbol{v},\Delta}$ for any $\boldsymbol{v}\in(\F_{q^2}^*)^n$ is $d_0\left(C_{\boldsymbol{v},\Delta}\right) =\min\{\Dis(\boldsymbol{e}) \mid \boldsymbol{e}\in \Delta\} = 28$ by \Cref{cdist}.

\begin{figure}[h]
    \centering
    \begin{tikzpicture}[y=0.7cm, x=0.7cm,font=\normalsize]
    \draw (0,0) -- (7,0);
    \draw (0,0) -- (0,5);

    \filldraw[fill=blue!40,draw=black!80] (0,0) circle (3pt)    node[anchor=south] {\scriptsize$48$};
    \filldraw[fill=blue!40,draw=black!80] (1,0) circle (3pt)    node[anchor=south] {\scriptsize$42$};
    \filldraw[fill=blue!40,draw=black!80] (2,0) circle (3pt)    node[anchor=south] {\scriptsize$36$};
    \filldraw[fill=black!40,draw=black!80] (3,0) circle (1pt)    node[anchor=south] {\scriptsize$30$};
    \filldraw[fill=black!40,draw=black!80] (4,0) circle (1pt)    node[anchor=south] {\scriptsize$24$};
    \filldraw[fill=black!40,draw=black!80] (5,0) circle (1pt)    node[anchor=south] {\scriptsize$18$};
    \filldraw[fill=black!40,draw=black!80] (6,0) circle (1pt)    node[anchor=south] {\scriptsize$12$};
    \filldraw[fill=black!40,draw=black!80] (7,0) circle (1pt)    node[anchor=south] {\scriptsize$6$};
    \filldraw[fill=blue!40,draw=black!80] (0,1) circle (3pt)    node[anchor=south] {\scriptsize$40$};
    \filldraw[fill=blue!40,draw=black!80] (1,1) circle (3pt)    node[anchor=south] {\scriptsize$35$};
    \filldraw[fill=blue!40,draw=black!80] (2,1) circle (3pt)    node[anchor=south] {\scriptsize$30$};
    \filldraw[fill=black!40,draw=black!80] (3,1) circle (1pt)    node[anchor=south] {\scriptsize$25$};
    \filldraw[fill=black!40,draw=black!80] (4,1) circle (1pt)    node[anchor=south] {\scriptsize$20$};
    \filldraw[fill=black!40,draw=black!80] (5,1) circle (1pt)    node[anchor=south] {\scriptsize$15$};
    \filldraw[fill=black!40,draw=black!80] (6,1) circle (1pt)    node[anchor=south] {\scriptsize$10$};
    \filldraw[fill=black!40,draw=black!80] (7,1) circle (1pt)    node[anchor=south] {\scriptsize$5$};
    \filldraw[fill=blue!40,draw=black!80] (0,2) circle (3pt)    node[anchor=south] {\scriptsize$32$};
    \filldraw[fill=blue!40,draw=black!80] (1,2) circle (3pt)    node[anchor=south] {\scriptsize$28$};
    \filldraw[fill=black!40,draw=black!80] (2,2) circle (1pt)    node[anchor=south] {\scriptsize$24$};
    \filldraw[fill=black!40,draw=black!80] (3,2) circle (1pt)    node[anchor=south] {\scriptsize$20$};
    \filldraw[fill=black!40,draw=black!80] (4,2) circle (1pt)    node[anchor=south] {\scriptsize$16$};
    \filldraw[fill=black!40,draw=black!80] (5,2) circle (1pt)    node[anchor=south] {\scriptsize$12$};
    \filldraw[fill=black!40,draw=black!80] (6,2) circle (1pt)    node[anchor=south] {\scriptsize$8$};
    \filldraw[fill=black!40,draw=black!80] (7,2) circle (1pt)    node[anchor=south] {\scriptsize$4$};
    \filldraw[fill=black!40,draw=black!80] (0,3) circle (1pt)    node[anchor=south] {\scriptsize$24$};
    \filldraw[fill=black!40,draw=black!80] (1,3) circle (1pt)    node[anchor=south] {\scriptsize$21$};
    \filldraw[fill=black!40,draw=black!80] (2,3) circle (1pt)    node[anchor=south] {\scriptsize$18$};
    \filldraw[fill=black!40,draw=black!80] (3,3) circle (1pt)    node[anchor=south] {\scriptsize$15$};
    \filldraw[fill=black!40,draw=black!80] (4,3) circle (1pt)    node[anchor=south] {\scriptsize$12$};
    \filldraw[fill=black!40,draw=black!80] (5,3) circle (1pt)    node[anchor=south] {\scriptsize$9$};
    \filldraw[fill=black!40,draw=black!80] (6,3) circle (1pt)    node[anchor=south] {\scriptsize$6$};
    \filldraw[fill=black!40,draw=black!80] (7,3) circle (1pt)    node[anchor=south] {\scriptsize$3$};
    \filldraw[fill=black!40,draw=black!80] (0,4) circle (1pt)    node[anchor=south] {\scriptsize$16$};
    \filldraw[fill=black!40,draw=black!80] (1,4) circle (1pt)    node[anchor=south] {\scriptsize$14$};
    \filldraw[fill=black!40,draw=black!80] (2,4) circle (1pt)    node[anchor=south] {\scriptsize$12$};
    \filldraw[fill=black!40,draw=black!80] (3,4) circle (1pt)    node[anchor=south] {\scriptsize$10$};
    \filldraw[fill=black!40,draw=black!80] (4,4) circle (1pt)    node[anchor=south] {\scriptsize$8$};
    \filldraw[fill=black!40,draw=black!80] (5,4) circle (1pt)    node[anchor=south] {\scriptsize$6$};
    \filldraw[fill=black!40,draw=black!80] (6,4) circle (1pt)    node[anchor=south] {\scriptsize$4$};
    \filldraw[fill=black!40,draw=black!80] (7,4) circle (1pt)    node[anchor=south] {\scriptsize$2$};
    \filldraw[fill=black!40,draw=black!80] (0,5) circle (1pt)    node[anchor=south] {\scriptsize$8$};
    \filldraw[fill=black!40,draw=black!80] (1,5) circle (1pt)    node[anchor=south] {\scriptsize$7$};
    \filldraw[fill=black!40,draw=black!80] (2,5) circle (1pt)    node[anchor=south] {\scriptsize$6$};
    \filldraw[fill=black!40,draw=black!80] (3,5) circle (1pt)    node[anchor=south] {\scriptsize$5$};
    \filldraw[fill=black!40,draw=black!80] (4,5) circle (1pt)    node[anchor=south] {\scriptsize$4$};
    \filldraw[fill=black!40,draw=black!80] (5,5) circle (1pt)    node[anchor=south] {\scriptsize$3$};
    \filldraw[fill=black!40,draw=black!80] (6,5) circle (1pt)    node[anchor=south] {\scriptsize$2$};
    \filldraw[fill=black!40,draw=black!80] (7,5) circle (1pt)    node[anchor=south] {\scriptsize$1$};

    \node [below] at (0,0) {\scriptsize$0$};
    \node [below] at (1,0) {\scriptsize$1$};
    \node [below] at (2,0) {\scriptsize$2$};
    \node [below] at (3,0) {\scriptsize$3$};
    \node [below] at (4,0) {\scriptsize$4$};
    \node [below] at (5,0) {\scriptsize$5$};
    \node [below] at (6,0) {\scriptsize$6$};
    \node [below] at (7,0) {\scriptsize$7$};
    \node [left] at (0,0) {\scriptsize$0$};
    \node [left] at (0,1) {\scriptsize$1$};
    \node [left] at (0,2) {\scriptsize$2$};
    \node [left] at (0,3) {\scriptsize$3$};
    \node [left] at (0,4) {\scriptsize$4$};
    \node [left] at (0,5) {\scriptsize$5$};

    \end{tikzpicture}
    \caption{Grid representation of $E$, where $m=2$, $a_1=8$, $a_2=6$ and $\Delta=\left(\{0,1,2\}\times\{0,1\}\right) \cup \{(0,2),(1,2)\}$.}
    \label{fig:grid}
\end{figure}
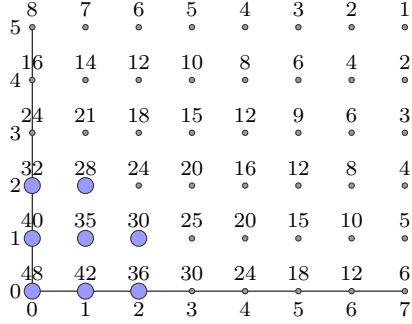
\end{example}

\begin{lemma}\label{le:isometricInnerProducts}
 Let $C_{\boldsymbol{v},\Delta}$ be a GMCC. Then  
 $(C_{\boldsymbol{v},\Delta})^{\perp_h}$ and $(C_{\boldsymbol{v},\Delta})^{\perp_e}$ are isometric.
\end{lemma}

\begin{proof}
It is straightforward because $(C_{\boldsymbol{v},\Delta})^{\perp_h}=((C_{\boldsymbol{v},\Delta})^{\perp_e})^q$.
\end{proof}

\begin{lemma}\label{le:isometricDual}
Let $C_{\boldsymbol{v},\Delta}$ be a GMCC. Then
$(C_{\boldsymbol{1},\Delta})^{\perp_h}$ and $(C_{\boldsymbol{v},\Delta})^{\perp_h}$ are isometric.
\end{lemma}

\begin{proof}
It follows from the fact that the family of GMCCs is closed under duality by \Cref{prop:ClosedByDuality} and by \Cref{le:isometricPrimal}.
\end{proof}

\begin{corollary}\label{re:TrickForDistance}
Let $C_{\boldsymbol{v},\Delta}$ be a GMCC. Then
$d((C_{\boldsymbol{v},\Delta})^{\perp_h})=d((C_{\boldsymbol{1},\Delta})^{\perp_e})$. 
\end{corollary}

\begin{proof}
 This is because $(C_{\boldsymbol{v},\Delta})^{\perp_h}$ and $(C_{\boldsymbol{1},\Delta})^{\perp_h}$ are isometric (by \Cref{le:isometricDual}) and also
 $(C_{\boldsymbol{1},\Delta})^{\perp_h}$ is isometric to $(C_{\boldsymbol{1},\Delta})^{\perp_e}$ (by \Cref{le:isometricInnerProducts}).
\end{proof}

\section{Stabilizer Quantum Codes From Generalized \\ Monomial-Cartesian Codes}\label{gen}

In the present section we construct stabilizer quantum codes by applying \Cref{th: quantherm} to GMMCs (\Cref{gmcc}) with a specific twist vector. Recall from \Cref{preliminaries} that $q$ is an odd prime power, $\zeta_{q^2-1}$ denotes a primitive $q^2-1$-th root of unity, $\lambda\in \N$ is such that $\lambda \mid q-1$,  $a_1= \lambda(q+1)$,  
$2 \leq a_j \leq q^2-1$ for all $j=2, \dots, m$, and $n=a_1a_2\cdots a_m$. 
We are going to choose the twist vector defined explicitly as follows:
\begin{equation}\label{twist1}
\boldsymbol{v}=(\underbrace{\zeta_{q^2-1}^{\frac{q-1}{2}},\dots ,\zeta_{q^2-1}^{\frac{q-1}{2}}}_{\frac{n}{q+1}}, \underbrace{1,\dots, 1}_{\frac{n}{q+1}}, \underbrace{\zeta_{q^2-1}^{\frac{q-1}{2}}, \dots ,\zeta_{q^2-1}^{\frac{q-1}{2}}}_{\frac{n}{q+1}}, \dots, \underbrace{1,\dots, 1}_{\frac{n}{q+1}})\in(\F_{q^2}^*)^n.
\end{equation}
By the forthcoming \Cref{lem: solutions}, it follows  that
\begin{equation*}
\boldsymbol{v}^{q+1} = (\underbrace{-1,  \dots, -1}_{\frac{n}{q+1}}, \underbrace{1,  \dots, 1}_{\frac{n}{q+1}}, \underbrace{-1,  \dots, -1}_{\frac{n}{q+1}}, \dots, \underbrace{1,  \dots, 1}_{\frac{n}{q+1}}).
\end{equation*}
Observe that there are $q+1$ blocks of $-1$'s or $1$'s. Recall that the coordinates $v_{\boldsymbol{\alpha}}$ of $\boldsymbol{v}$ are labelled and ordered in the same way as the points $\boldsymbol{P_\alpha}\in Z$. This twist vector works as follows. For each $\boldsymbol{\alpha}\in E$,
\begin{equation}\label{vq+1}
        v_{\boldsymbol{\alpha}}^{q+1}=
		  \begin{cases}
			 -1 & \text{ if } 0 \leq (\alpha_1 \bmod 2\lambda ) \leq \lambda-1,\\
			 1  &  \text{ if }  \lambda \leq (\alpha_1 \bmod 2\lambda ) \leq 2\lambda-1.
		  \end{cases}
    \end{equation}
Notice that $v_{\boldsymbol{\alpha}}$ only depends on $\alpha_1$. 
The reason why we choose this specific twist vector is going to become clear in 
\Cref{le:generallambda}.
Notice that the existence of these $v_{\boldsymbol{\alpha}}^{q+1}=\left(\zeta_{q^2-1}^{\frac{q-1}{2}}\right)^{q+1}=-1$ for some primitive element $\zeta_{q^2-1}\in\F_{q^2}^*$ is guaranteed by the next lemma:

\begin{lemma}\label{lem: solutions}
    Let $q$ be an odd prime power. Then $x^{q+1}=-1$ has $q+1$ solutions in $\F_{q^2}$.
\end{lemma}

\begin{proof}
    For a primitive element $\zeta_{q^2-1} \in \F_{q^2}^*$,
    we  consider $\zeta_{q^2-1}^{\frac{q-1}{2}}$. We have that $$\left(\zeta_{q^2-1}^{\frac{q-1}{2}}\right)^{q+1}= \zeta_{q^2-1}^{\frac{(q+1)(q-1)}{2}}
    = \zeta_{q^2-1}^{\frac{q^2-1}{2}}=-1.$$
    The same will be true for $\zeta_{q+1}^k \zeta_{q^2-1}^{\frac{q-1}{2}}$ where $\zeta_{q+1}^k$, $0\leq k \leq q$ is
    any $q+1$-th root of unity.\qedhere
\end{proof}

%\begin{remark}
%We have a corresponding twist vector in one variable
%\[
%    \boldsymbol{v}=
%\]
%such that 
%\begin{equation}
%    \boldsymbol{v}^{q+1}\coloneqq  (\underbrace{\pm1,  \dots, \pm1}_{\lambda}, \underbrace{\mp1,  \dots, \mp1}_{\lambda}, \underbrace{\pm1,  \dots, \pm1}_{\lambda}, \dots, \underbrace{\mp1,  \dots, \mp1}_{\lambda}).
%\end{equation}
%\end{remark}

\subsection{Self-Orthogonality Conditions}

First we present some conditions for the evaluation vectors of monomials in $R$ to
be orthogonal for the Hermitian inner product, when our twist vector is used.

\begin{proposition}\label{le:generallambda}
    Keep the same notations as before. Let $q$ be an odd prime power and consider the  twist vector $\boldsymbol{v}$ defined in \eqref{twist1}. Let  $\boldsymbol{e}=(e_1, \dots, e_m)$, $\boldsymbol{e'}=(e'_1, \dots, e'_m) \in E$ be exponents of two monomials $X^{\boldsymbol{e}}$, $X^{\boldsymbol{e'}} \in R$. Then, the evaluation vectors under the map $\ev_{\boldsymbol{v}}$ of these monomials are orthogonal for the Hermitian inner product if one of the following conditions hold:
    \begin{itemize}
        \item $e_1 \equiv e'_1 \mod q+1$, or
        \item $e_1 \not\equiv e'_1 \mod \frac{q+1}{2}$.
    \end{itemize}
\end{proposition}

\begin{proof}
    In order to compute some conditions under which two evaluations of monomials of the quotient ring $R$ are orthogonal for the Hermitian inner product we have to see when the following sum vanishes:
    \[
     \ev_{\boldsymbol{v}}(X^{\boldsymbol{e}}) \cdot_h \ev_{\boldsymbol{v}}(X^{\boldsymbol{e'}})= \sum_{\boldsymbol{\alpha}\in E} v_{\boldsymbol{\alpha}}^{q+1} \zeta_{\lambda(q+1)}^{\alpha_1(e_1 + q e'_1)} \xi_{(2,\alpha_2)}^{(e_2 +q e'_2)} \cdots \xi_{(m,\alpha_m)}^{(e_m + q e'_m)}.
    \] 
    Since $v_{\boldsymbol{\alpha}}$ only depends on $\alpha_1$ we can denote by $v_{\alpha_1}\coloneqq v_{(\alpha_1, \dots, \alpha_m)}=v_{\boldsymbol{\alpha}}$ and reorder the above sum in the following way:
\[
    \ev_{\boldsymbol{v}}(X^{\boldsymbol{e}})\ \cdot_h \ \ev_{\boldsymbol{v}}(X^{\boldsymbol{e'}}) = \left( \sum_{\alpha_1=0}^{\lambda (q+1)-1} v_{\alpha_1}^{q+1} \zeta_{\lambda(q+1)}^{\alpha_1(e_1 + q e'_1)} \right)\left(\sum_{\alpha_2=0}^{a_2-1} \xi_{(2,\alpha_2)}^{(e_2 +q e'_2)}\right) \dots \left(\sum_{\alpha_m=0}^{a_m-1} \xi_{(m,\alpha_m)}^{(e_m + q e'_m)}\right).
    \]
We can do that because all the coordinates $v_{\boldsymbol{\alpha}}$ in $\boldsymbol{v}$ that have the same $\alpha_1$ have the same value. Now we study when the first factor equals 0, and we will ignore the other factors, since the first one gives enough information for the proof. Consider then
\begin{equation}\label{S}
\sum_{\alpha_1=0}^{\lambda (q+1)-1} v_{\alpha_1}^{q+1} \zeta_{\lambda(q+1)}^{\alpha_1(e_1 + q e'_1)},
\end{equation}
which is a sum over $\alpha_1\in\{0,1,\dots,\lambda (q+1)-1\}$. We write each $\alpha_1$ in the form $k\lambda +r$ where $0\le k \le q$ and $0\le r <\lambda$.
Using this to break \eqref{S} into $\lambda$ blocks of size $q+1$, using the fact that $\zeta_{q+1}:=\zeta_{\lambda(q+1)}^\lambda$ is a primitive $q+1$-th root of unity and using the structure of the twist vector $\boldsymbol{v}$, we can write \eqref{S} as

\begin{align*}
    \sum_{\alpha_1=0}^{\lambda (q+1)-1} v_{\alpha_1}^{q+1} \zeta_{\lambda(q+1)}^{\alpha_1(e_1 + q e'_1)} &=  \sum_{\substack{0\leq k\leq q \\ 0 \leq r < \lambda}} v_{k\lambda+r}^{q+1} \zeta_{\lambda(q+1)}^{(k\lambda+r)(e_1+qe'_1)} = \sum_{k=0}^q v_{k\lambda}^{q+1} \zeta_{q+1}^{k(e_1+qe'_1)} \\
    & + \zeta_{\lambda(q+1)}^{e_1+qe'_1} \sum_{k=0}^q v_{k\lambda+1}^{q+1} \zeta_{q+1}^{k(e_1+qe'_1)} + \cdots + \zeta_{\lambda(q+1)}^{(\lambda-1)(e_1+qe'_1)} \sum_{k=0}^q v_{k\lambda+\lambda-1}^{q+1} \zeta_{q+1}^{k(e_1+qe'_1)} \\
    & =\left(1 + \zeta_{\lambda(q+1)}^{(e_1+qe'_1)} + \dots + \zeta_{\lambda(q+1)}^{(\lambda-1)(e_1+qe'_1)}\right) \left(\sum_{k=0}^{q} v_{k\lambda}^{q+1} \zeta_{q+1}^{k(e_1 + q e'_1)}\right).
\end{align*}

Notice that we can do that because from \eqref{vq+1} and the fact that $1\leq\lambda\leq q-1$ we have that $v_{k\lambda}^{q+1}=v_{k\lambda+1}^{q+1}=\cdots=v_{k\lambda+\lambda-1}^{q+1}$ for all $0\leq k \leq q$. Now using again \eqref{vq+1} and the fact that $\zeta_{\frac{q+1}{2}}:=\zeta_{q+1}^2$ is a primitive $\frac{q+1}{2}$-th root of unity, we rewrite the last sum in the following way:
\begin{align*}
    \sum_{k=0}^{q} v_{k\lambda}^{q+1} \zeta_{q+1}^{k(e_1 + q e'_1)} &= \sum_{k=0}^{\frac{q-1}{2}} v_{2k\lambda}^{q+1} \zeta_{q+1}^{2k(e_1 + q e'_1)}+\sum_{k=0}^{\frac{q-1}{2}} v_{2k\lambda+1}^{q+1} \zeta_{q+1}^{(2k+1)(e_1 + q e'_1)}\\
    &=\sum_{k=0}^{\frac{q-1}{2}} v_{2k\lambda}^{q+1} \zeta_{q+1}^{2k(e_1 + q e'_1)}-\zeta_{q+1}^{e_1+qe'_1} \sum_{k=0}^{\frac{q-1}{2}} v_{2k\lambda}^{q+1} \zeta_{q+1}^{2k(e_1 + q e'_1)}\\
    &=\zeta_{q+1}^{e_1+qe'_1}\left(\sum_{k=0}^{\frac{q-1}{2}} \zeta_{\frac{q+1}{2}}^{k(e_1 + q e'_1)}\right) -\left(\sum_{k=0}^{\frac{q-1}{2}} \zeta_{\frac{q+1}{2}}^{k(e_1 + q e'_1)}\right) \\
    &= (\zeta_{q+1}^{e_1+qe'_1}-1)\left(\sum_{k=0}^{\frac{q-1}{2}} \zeta_{\frac{q+1}{2}}^{k(e_1 + q e'_1)}\right).
\end{align*}

    Thus, we have shown that we can write \eqref{S} as
    \begin{equation*}
    \sum_{\alpha_1=0}^{\lambda (q+1)-1} v_{\alpha_1}^{q+1} \zeta_{\lambda(q+1)}^{\alpha_1(e_1 + q e'_1)}=P\left(\zeta_{\lambda(q+1)}^{e_1+qe'_1}\right)
    \biggl(\zeta_{q+1}^{e_1+qe'_1}-1\biggr)\left(\sum_{k=0}^{\frac{q-1}{2}} \zeta_{\frac{q+1}{2}}^{k(e_1 + q e'_1)}\right),
    \end{equation*}
    where
    $P(x)= 1+x+x^2+ \dots + x^{\lambda -1}$.
    The above product equals 0 if and only if one of the following conditions holds:
    \begin{itemize}
        \item $\zeta_{q+1}^{e_1+qe'_1} -1 = 0$ $\iff$ $e_1+qe'_1 \equiv 0 \mod q+1$. That is, $e_1 \equiv e'_1 \mod q+1$; or
        \item $\left(\sum_{k=0}^{\frac{q-1}{2}} \zeta_{\frac{q+1}{2}}^{k(e_1 + q e'_1)}\right)=0$ $\iff$ $e_1+qe'_1 \not\equiv 0 \mod \frac{q+1}{2}$. Since $q\equiv -1 \mod \frac{q+1}{2}$, this is equivalent to $e_1 \not\equiv e'_1 \mod \frac{q+1}{2}$; or
        \item $P\left(\zeta_{\lambda(q+1)}^{(e_1+qe'_1)}\right)=0$. This is true if and only if $\zeta_{\lambda(q+1)}^{(e_1+qe'_1)}$ is a $\lambda$-th root of unity other than 1. That is equivalent to $e_1+qe'_1 \equiv 0 \mod q+1$ and $e_1+qe'_1\not\equiv 0 \mod \lambda(q+1)$, which is a particular case of the first condition.
    \end{itemize} 
    Therefore, if either of the first two conditions hold, the sum
    \eqref{S} equals 0 and that implies that $\ev_{\boldsymbol{v}}(X^{\boldsymbol{e}})$ and $\ev_{\boldsymbol{v}}(X^{\boldsymbol{e'}})$ are orthogonal for the Hermitian inner product.
\end{proof}

\begin{remark}
    Consider the case when the twist vector is $\boldsymbol{1}$, $\lambda=1$ and $A_j$ is the set of $q+1$-th roots of unity, that is the solutions to $X_j^{q+1}-1=0$, for every $j=1,\dots,m$. Then for any $\Delta\subseteq E$ the GMCC $C_{\boldsymbol{1},\Delta}$  is an Affine Variety Code (AVC) and it is not self-orthogonal (for the Hermitian inner product). This is because when we compute the Hermitian inner product of the evaluations of any monomial $X^{\boldsymbol{e}}=X^{(e_1,\dots,e_m)}$ with itself, one obtains that 
    \begin{align*}
        \ev_{\boldsymbol{1}}(X^{\boldsymbol{e}}) \cdot_h \ev_{\boldsymbol{1}}(X^{\boldsymbol{e}}) & = \sum_{\boldsymbol{\alpha}\in E}  \zeta_{q+1}^{\alpha_1e_1 (1+ q)} \zeta_{q+1}^{\alpha_2e_2 (1+ q)} \cdots \zeta_{q+1}^{\alpha_me_m (1+ q)}
        \\ &= \left( \sum_{\alpha_1=0}^{q} \zeta_{q+1}^{\alpha_1e_1(1 + q)} \right)\left(\sum_{\alpha_2=0}^{q} \zeta_{q+1}^{\alpha_2e_2(1 +q)}\right) \dots \left(\sum_{\alpha_m=0}^{q} \zeta_{q+1}^{\alpha_me_m(1 + q)}\right) 
    \end{align*}
    and every factor above is
    \[
    \sum_{k=0}^q \zeta_{q+1}^{ke_1(1 + q)} =q+1 \neq 0.
    \]
    Thus, 
    the evaluation of a monomial is not orthogonal to itself, and these codes are not self-orthogonal. However, we are able to provide a twist vector $\boldsymbol{v}$ \eqref{twist1} to construct a self-orthogonal GMCC $C_{\boldsymbol{v},\Delta}$ which is isometric to the non self-orthogonal AVC $C_{\boldsymbol{1},\Delta}$.
    The problem of not getting evaluations of monomials to be self-orthogonal can happen also with other twist vectors, that is why one has to choose the twist vector carefully.
\end{remark}

\subsection{Our General Construction}

Before stating the theorem that is the general construction of this paper, 
recall the definition of the set $E$ in the previous section.
We define a subset in $E$ which will be useful in the following.

\begin{definition}\label{Delta0}
    Let $E_0:=\left\{\boldsymbol{e}=(e_1,\dots,e_m)\in E \mid 0 \leq e_1\leq\frac{q-1}{2}\right\}\subseteq E$.
\end{definition}

The next theorem shows that the set $E_0$ introduced in \Cref{Delta0} is used as a reference to construct Hermitian self-orthogonal GMCCs.

\begin{theorem}\label{te:SelfOrthogonal2}
    Let $q$ be an odd prime power and let $m\geq 1$, $\lambda \mid q-1$, $a_1:=\lambda(q+1)$ and $2\leq a_j \leq q^2-1$, $j=2,\dots,m$ be positive integers.
    Let $n:=a_1\cdots a_m$. Consider the twist vector $\boldsymbol{v}$ defined in \eqref{twist1} and the set $E_0\subseteq E$ introduced in \Cref{Delta0}. Let $\Delta$ be a subset of $E_0$. Then, 
    $$C_{\boldsymbol{v},\Delta}\subseteq (C_{\boldsymbol{v},\Delta})^{\perp_h}.$$ Therefore, there exists a stabilizer quantum code with parameters
    $$[[n,n-2\#\Delta, \geq d]]_q$$
    where $d=\dis((C_{\boldsymbol{1},\Delta})^{\perp_e})$.
\end{theorem}

\begin{proof}
    Since for all $(e_1,\dots,e_m)\in \Delta$ we have $e_1\leq \frac{q-1}{2}$, the self-orthogonality follows from \Cref{le:generallambda}.
    The existence and parameters of the stabilizer quantum code follows from \Cref{th: quantherm}. Notice that $d=\dis((C_{\boldsymbol{v},\Delta})^{\perp_h})$, but from \Cref{re:TrickForDistance} we can conclude that $d=\dis((C_{\boldsymbol{1},\Delta})^{\perp_e})$.
\end{proof}

Notice that in the above theorem we do not give an explicit bound for the minimum distance, but it can be computed using \Cref{cdist} in every particular case.

\subsection{Our Specific Construction}

Now we are going to provide a strategy \cite{HoholdtHyperbolicCodes} to choose a set $\Delta\subseteq E_0$ so that we can control the minimum distance $\dis((C_{\boldsymbol{1},\Delta})^{\perp_e})$ and it maximizes the dimension of the resulting stabilizer quantum code.
To that purpose, we need the following
\begin{definition}\label{Deltat}
Let $2\leq t \leq \frac{q+3}{2}$ be a positive integer. Define
$$\Delta_t:=\left\{\boldsymbol{e}=(e_1,\ldots,e_m)\in E \ \middle\vert \ \prod_{j=1}^m (e_j+1)<t\right\}\subseteq E.$$
\end{definition}

\noindent Some instances of the above set are represented in \Cref{fig:Deltat}.

\begin{figure}[h]
\centering
    \begin{multicols}{3}
    \centering
    \begin{subfigure}[b]{0.3\textwidth}
    \centering
    \begin{tikzpicture}[y=0.6cm, x=0.6cm,font=\normalsize]
    \draw (0,0) -- (7,0);
    \draw (0,0) -- (0,5);

    \filldraw[fill=blue!40,draw=black!80] (0,0) circle (3pt)    node[anchor=south] {\scriptsize$48$};
    \filldraw[fill=blue!40,draw=black!80] (1,0) circle (3pt)    node[anchor=south] {\scriptsize$42$};
    \filldraw[fill=black!40,draw=black!80] (2,0) circle (1pt)    node[anchor=south] {\scriptsize$36$};
    \filldraw[fill=black!40,draw=black!80] (3,0) circle (1pt)    node[anchor=south] {\scriptsize$30$};
    \filldraw[fill=black!40,draw=black!80] (4,0) circle (1pt)    node[anchor=south] {\scriptsize$24$};
    \filldraw[fill=black!40,draw=black!80] (5,0) circle (1pt)    node[anchor=south] {\scriptsize$18$};
    \filldraw[fill=black!40,draw=black!80] (6,0) circle (1pt)    node[anchor=south] {\scriptsize$12$};
    \filldraw[fill=black!40,draw=black!80] (7,0) circle (1pt)    node[anchor=south] {\scriptsize$6$};
    \filldraw[fill=blue!40,draw=black!80] (0,1) circle (3pt)    node[anchor=south] {\scriptsize$40$};
    \filldraw[fill=black!40,draw=black!80] (1,1) circle (1pt)    node[anchor=south] {\scriptsize$35$};
    \filldraw[fill=black!40,draw=black!80] (2,1) circle (1pt)    node[anchor=south] {\scriptsize$30$};
    \filldraw[fill=black!40,draw=black!80] (3,1) circle (1pt)    node[anchor=south] {\scriptsize$25$};
    \filldraw[fill=black!40,draw=black!80] (4,1) circle (1pt)    node[anchor=south] {\scriptsize$20$};
    \filldraw[fill=black!40,draw=black!80] (5,1) circle (1pt)    node[anchor=south] {\scriptsize$15$};
    \filldraw[fill=black!40,draw=black!80] (6,1) circle (1pt)    node[anchor=south] {\scriptsize$10$};
    \filldraw[fill=black!40,draw=black!80] (7,1) circle (1pt)    node[anchor=south] {\scriptsize$5$};
    \filldraw[fill=black!40,draw=black!80] (0,2) circle (1pt)    node[anchor=south] {\scriptsize$32$};
    \filldraw[fill=black!40,draw=black!80] (1,2) circle (1pt)    node[anchor=south] {\scriptsize$28$};
    \filldraw[fill=black!40,draw=black!80] (2,2) circle (1pt)    node[anchor=south] {\scriptsize$24$};
    \filldraw[fill=black!40,draw=black!80] (3,2) circle (1pt)    node[anchor=south] {\scriptsize$20$};
    \filldraw[fill=black!40,draw=black!80] (4,2) circle (1pt)    node[anchor=south] {\scriptsize$16$};
    \filldraw[fill=black!40,draw=black!80] (5,2) circle (1pt)    node[anchor=south] {\scriptsize$12$};
    \filldraw[fill=black!40,draw=black!80] (6,2) circle (1pt)    node[anchor=south] {\scriptsize$8$};
    \filldraw[fill=black!40,draw=black!80] (7,2) circle (1pt)    node[anchor=south] {\scriptsize$4$};
    \filldraw[fill=black!40,draw=black!80] (0,3) circle (1pt)    node[anchor=south] {\scriptsize$24$};
    \filldraw[fill=black!40,draw=black!80] (1,3) circle (1pt)    node[anchor=south] {\scriptsize$21$};
    \filldraw[fill=black!40,draw=black!80] (2,3) circle (1pt)    node[anchor=south] {\scriptsize$18$};
    \filldraw[fill=black!40,draw=black!80] (3,3) circle (1pt)    node[anchor=south] {\scriptsize$15$};
    \filldraw[fill=black!40,draw=black!80] (4,3) circle (1pt)    node[anchor=south] {\scriptsize$12$};
    \filldraw[fill=black!40,draw=black!80] (5,3) circle (1pt)    node[anchor=south] {\scriptsize$9$};
    \filldraw[fill=black!40,draw=black!80] (6,3) circle (1pt)    node[anchor=south] {\scriptsize$6$};
    \filldraw[fill=black!40,draw=black!80] (7,3) circle (1pt)    node[anchor=south] {\scriptsize$3$};
    \filldraw[fill=black!40,draw=black!80] (0,4) circle (1pt)    node[anchor=south] {\scriptsize$16$};
    \filldraw[fill=black!40,draw=black!80] (1,4) circle (1pt)    node[anchor=south] {\scriptsize$14$};
    \filldraw[fill=black!40,draw=black!80] (2,4) circle (1pt)    node[anchor=south] {\scriptsize$12$};
    \filldraw[fill=black!40,draw=black!80] (3,4) circle (1pt)    node[anchor=south] {\scriptsize$10$};
    \filldraw[fill=black!40,draw=black!80] (4,4) circle (1pt)    node[anchor=south] {\scriptsize$8$};
    \filldraw[fill=black!40,draw=black!80] (5,4) circle (1pt)    node[anchor=south] {\scriptsize$6$};
    \filldraw[fill=black!40,draw=black!80] (6,4) circle (1pt)    node[anchor=south] {\scriptsize$4$};
    \filldraw[fill=black!40,draw=black!80] (7,4) circle (1pt)    node[anchor=south] {\scriptsize$2$};
    \filldraw[fill=black!40,draw=black!80] (0,5) circle (1pt)    node[anchor=south] {\scriptsize$8$};
    \filldraw[fill=black!40,draw=black!80] (1,5) circle (1pt)    node[anchor=south] {\scriptsize$7$};
    \filldraw[fill=black!40,draw=black!80] (2,5) circle (1pt)    node[anchor=south] {\scriptsize$6$};
    \filldraw[fill=black!40,draw=black!80] (3,5) circle (1pt)    node[anchor=south] {\scriptsize$5$};
    \filldraw[fill=black!40,draw=black!80] (4,5) circle (1pt)    node[anchor=south] {\scriptsize$4$};
    \filldraw[fill=black!40,draw=black!80] (5,5) circle (1pt)    node[anchor=south] {\scriptsize$3$};
    \filldraw[fill=black!40,draw=black!80] (6,5) circle (1pt)    node[anchor=south] {\scriptsize$2$};
    \filldraw[fill=black!40,draw=black!80] (7,5) circle (1pt)    node[anchor=south] {\scriptsize$1$};

    \node [below] at (0,0) {\scriptsize$0$};
    \node [below] at (1,0) {\scriptsize$1$};
    \node [below] at (2,0) {\scriptsize$2$};
    \node [below] at (3,0) {\scriptsize$3$};
    \node [below] at (4,0) {\scriptsize$4$};
    \node [below] at (5,0) {\scriptsize$5$};
    \node [below] at (6,0) {\scriptsize$6$};
    \node [below] at (7,0) {\scriptsize$7$};
    \node [left] at (0,0) {\scriptsize$0$};
    \node [left] at (0,1) {\scriptsize$1$};
    \node [left] at (0,2) {\scriptsize$2$};
    \node [left] at (0,3) {\scriptsize$3$};
    \node [left] at (0,4) {\scriptsize$4$};
    \node [left] at (0,5) {\scriptsize$5$};

    \end{tikzpicture}
    \caption{$\Delta_3$}
    \end{subfigure}
    \columnbreak

    \centering
        \begin{subfigure}[b]{0.3\textwidth}
    \centering
    \begin{tikzpicture}[y=0.6cm, x=0.6cm,font=\normalsize]
    \draw (0,0) -- (7,0);
    \draw (0,0) -- (0,5);

    \filldraw[fill=blue!40,draw=black!80] (0,0) circle (3pt)    node[anchor=south] {\scriptsize$48$};
    \filldraw[fill=blue!40,draw=black!80] (1,0) circle (3pt)    node[anchor=south] {\scriptsize$42$};
    \filldraw[fill=blue!40,draw=black!80] (2,0) circle (3pt)    node[anchor=south] {\scriptsize$36$};
    \filldraw[fill=black!40,draw=black!80] (3,0) circle (1pt)    node[anchor=south] {\scriptsize$30$};
    \filldraw[fill=black!40,draw=black!80] (4,0) circle (1pt)    node[anchor=south] {\scriptsize$24$};
    \filldraw[fill=black!40,draw=black!80] (5,0) circle (1pt)    node[anchor=south] {\scriptsize$18$};
    \filldraw[fill=black!40,draw=black!80] (6,0) circle (1pt)    node[anchor=south] {\scriptsize$12$};
    \filldraw[fill=black!40,draw=black!80] (7,0) circle (1pt)    node[anchor=south] {\scriptsize$6$};
    \filldraw[fill=blue!40,draw=black!80] (0,1) circle (3pt)    node[anchor=south] {\scriptsize$40$};
    \filldraw[fill=black!40,draw=black!80] (1,1) circle (1pt)    node[anchor=south] {\scriptsize$35$};
    \filldraw[fill=black!40,draw=black!80] (2,1) circle (1pt)    node[anchor=south] {\scriptsize$30$};
    \filldraw[fill=black!40,draw=black!80] (3,1) circle (1pt)    node[anchor=south] {\scriptsize$25$};
    \filldraw[fill=black!40,draw=black!80] (4,1) circle (1pt)    node[anchor=south] {\scriptsize$20$};
    \filldraw[fill=black!40,draw=black!80] (5,1) circle (1pt)    node[anchor=south] {\scriptsize$15$};
    \filldraw[fill=black!40,draw=black!80] (6,1) circle (1pt)    node[anchor=south] {\scriptsize$10$};
    \filldraw[fill=black!40,draw=black!80] (7,1) circle (1pt)    node[anchor=south] {\scriptsize$5$};
    \filldraw[fill=blue!40,draw=black!80] (0,2) circle (3pt)    node[anchor=south] {\scriptsize$32$};
    \filldraw[fill=black!40,draw=black!80] (1,2) circle (1pt)    node[anchor=south] {\scriptsize$28$};
    \filldraw[fill=black!40,draw=black!80] (2,2) circle (1pt)    node[anchor=south] {\scriptsize$24$};
    \filldraw[fill=black!40,draw=black!80] (3,2) circle (1pt)    node[anchor=south] {\scriptsize$20$};
    \filldraw[fill=black!40,draw=black!80] (4,2) circle (1pt)    node[anchor=south] {\scriptsize$16$};
    \filldraw[fill=black!40,draw=black!80] (5,2) circle (1pt)    node[anchor=south] {\scriptsize$12$};
    \filldraw[fill=black!40,draw=black!80] (6,2) circle (1pt)    node[anchor=south] {\scriptsize$8$};
    \filldraw[fill=black!40,draw=black!80] (7,2) circle (1pt)    node[anchor=south] {\scriptsize$4$};
    \filldraw[fill=black!40,draw=black!80] (0,3) circle (1pt)    node[anchor=south] {\scriptsize$24$};
    \filldraw[fill=black!40,draw=black!80] (1,3) circle (1pt)    node[anchor=south] {\scriptsize$21$};
    \filldraw[fill=black!40,draw=black!80] (2,3) circle (1pt)    node[anchor=south] {\scriptsize$18$};
    \filldraw[fill=black!40,draw=black!80] (3,3) circle (1pt)    node[anchor=south] {\scriptsize$15$};
    \filldraw[fill=black!40,draw=black!80] (4,3) circle (1pt)    node[anchor=south] {\scriptsize$12$};
    \filldraw[fill=black!40,draw=black!80] (5,3) circle (1pt)    node[anchor=south] {\scriptsize$9$};
    \filldraw[fill=black!40,draw=black!80] (6,3) circle (1pt)    node[anchor=south] {\scriptsize$6$};
    \filldraw[fill=black!40,draw=black!80] (7,3) circle (1pt)    node[anchor=south] {\scriptsize$3$};
    \filldraw[fill=black!40,draw=black!80] (0,4) circle (1pt)    node[anchor=south] {\scriptsize$16$};
    \filldraw[fill=black!40,draw=black!80] (1,4) circle (1pt)    node[anchor=south] {\scriptsize$14$};
    \filldraw[fill=black!40,draw=black!80] (2,4) circle (1pt)    node[anchor=south] {\scriptsize$12$};
    \filldraw[fill=black!40,draw=black!80] (3,4) circle (1pt)    node[anchor=south] {\scriptsize$10$};
    \filldraw[fill=black!40,draw=black!80] (4,4) circle (1pt)    node[anchor=south] {\scriptsize$8$};
    \filldraw[fill=black!40,draw=black!80] (5,4) circle (1pt)    node[anchor=south] {\scriptsize$6$};
    \filldraw[fill=black!40,draw=black!80] (6,4) circle (1pt)    node[anchor=south] {\scriptsize$4$};
    \filldraw[fill=black!40,draw=black!80] (7,4) circle (1pt)    node[anchor=south] {\scriptsize$2$};
    \filldraw[fill=black!40,draw=black!80] (0,5) circle (1pt)    node[anchor=south] {\scriptsize$8$};
    \filldraw[fill=black!40,draw=black!80] (1,5) circle (1pt)    node[anchor=south] {\scriptsize$7$};
    \filldraw[fill=black!40,draw=black!80] (2,5) circle (1pt)    node[anchor=south] {\scriptsize$6$};
    \filldraw[fill=black!40,draw=black!80] (3,5) circle (1pt)    node[anchor=south] {\scriptsize$5$};
    \filldraw[fill=black!40,draw=black!80] (4,5) circle (1pt)    node[anchor=south] {\scriptsize$4$};
    \filldraw[fill=black!40,draw=black!80] (5,5) circle (1pt)    node[anchor=south] {\scriptsize$3$};
    \filldraw[fill=black!40,draw=black!80] (6,5) circle (1pt)    node[anchor=south] {\scriptsize$2$};
    \filldraw[fill=black!40,draw=black!80] (7,5) circle (1pt)    node[anchor=south] {\scriptsize$1$};

    \node [below] at (0,0) {\scriptsize$0$};
    \node [below] at (1,0) {\scriptsize$1$};
    \node [below] at (2,0) {\scriptsize$2$};
    \node [below] at (3,0) {\scriptsize$3$};
    \node [below] at (4,0) {\scriptsize$4$};
    \node [below] at (5,0) {\scriptsize$5$};
    \node [below] at (6,0) {\scriptsize$6$};
    \node [below] at (7,0) {\scriptsize$7$};
    \node [left] at (0,0) {\scriptsize$0$};
    \node [left] at (0,1) {\scriptsize$1$};
    \node [left] at (0,2) {\scriptsize$2$};
    \node [left] at (0,3) {\scriptsize$3$};
    \node [left] at (0,4) {\scriptsize$4$};
    \node [left] at (0,5) {\scriptsize$5$};

    \end{tikzpicture}
    \caption{$\Delta_4$}
    \end{subfigure}

        \centering
        \begin{subfigure}[b]{0.3\textwidth}
    \centering
    \begin{tikzpicture}[y=0.6cm, x=0.6cm,font=\normalsize]
    \draw (0,0) -- (7,0);
    \draw (0,0) -- (0,5);

    \filldraw[fill=blue!40,draw=black!80] (0,0) circle (3pt)    node[anchor=south] {\scriptsize$48$};
    \filldraw[fill=blue!40,draw=black!80] (1,0) circle (3pt)    node[anchor=south] {\scriptsize$42$};
    \filldraw[fill=blue!40,draw=black!80] (2,0) circle (3pt)    node[anchor=south] {\scriptsize$36$};
    \filldraw[fill=blue!40,draw=black!80] (3,0) circle (3pt)    node[anchor=south] {\scriptsize$30$};
    \filldraw[fill=black!40,draw=black!80] (4,0) circle (1pt)    node[anchor=south] {\scriptsize$24$};
    \filldraw[fill=black!40,draw=black!80] (5,0) circle (1pt)    node[anchor=south] {\scriptsize$18$};
    \filldraw[fill=black!40,draw=black!80] (6,0) circle (1pt)    node[anchor=south] {\scriptsize$12$};
    \filldraw[fill=black!40,draw=black!80] (7,0) circle (1pt)    node[anchor=south] {\scriptsize$6$};
    \filldraw[fill=blue!40,draw=black!80] (0,1) circle (3pt)    node[anchor=south] {\scriptsize$40$};
    \filldraw[fill=blue!40,draw=black!80] (1,1) circle (3pt)    node[anchor=south] {\scriptsize$35$};
    \filldraw[fill=black!40,draw=black!80] (2,1) circle (1pt)    node[anchor=south] {\scriptsize$30$};
    \filldraw[fill=black!40,draw=black!80] (3,1) circle (1pt)    node[anchor=south] {\scriptsize$25$};
    \filldraw[fill=black!40,draw=black!80] (4,1) circle (1pt)    node[anchor=south] {\scriptsize$20$};
    \filldraw[fill=black!40,draw=black!80] (5,1) circle (1pt)    node[anchor=south] {\scriptsize$15$};
    \filldraw[fill=black!40,draw=black!80] (6,1) circle (1pt)    node[anchor=south] {\scriptsize$10$};
    \filldraw[fill=black!40,draw=black!80] (7,1) circle (1pt)    node[anchor=south] {\scriptsize$5$};
    \filldraw[fill=blue!40,draw=black!80] (0,2) circle (3pt)    node[anchor=south] {\scriptsize$32$};
    \filldraw[fill=black!40,draw=black!80] (1,2) circle (1pt)    node[anchor=south] {\scriptsize$28$};
    \filldraw[fill=black!40,draw=black!80] (2,2) circle (1pt)    node[anchor=south] {\scriptsize$24$};
    \filldraw[fill=black!40,draw=black!80] (3,2) circle (1pt)    node[anchor=south] {\scriptsize$20$};
    \filldraw[fill=black!40,draw=black!80] (4,2) circle (1pt)    node[anchor=south] {\scriptsize$16$};
    \filldraw[fill=black!40,draw=black!80] (5,2) circle (1pt)    node[anchor=south] {\scriptsize$12$};
    \filldraw[fill=black!40,draw=black!80] (6,2) circle (1pt)    node[anchor=south] {\scriptsize$8$};
    \filldraw[fill=black!40,draw=black!80] (7,2) circle (1pt)    node[anchor=south] {\scriptsize$4$};
    \filldraw[fill=blue!40,draw=black!80] (0,3) circle (3pt)    node[anchor=south] {\scriptsize$24$};
    \filldraw[fill=black!40,draw=black!80] (1,3) circle (1pt)    node[anchor=south] {\scriptsize$21$};
    \filldraw[fill=black!40,draw=black!80] (2,3) circle (1pt)    node[anchor=south] {\scriptsize$18$};
    \filldraw[fill=black!40,draw=black!80] (3,3) circle (1pt)    node[anchor=south] {\scriptsize$15$};
    \filldraw[fill=black!40,draw=black!80] (4,3) circle (1pt)    node[anchor=south] {\scriptsize$12$};
    \filldraw[fill=black!40,draw=black!80] (5,3) circle (1pt)    node[anchor=south] {\scriptsize$9$};
    \filldraw[fill=black!40,draw=black!80] (6,3) circle (1pt)    node[anchor=south] {\scriptsize$6$};
    \filldraw[fill=black!40,draw=black!80] (7,3) circle (1pt)    node[anchor=south] {\scriptsize$3$};
    \filldraw[fill=black!40,draw=black!80] (0,4) circle (1pt)    node[anchor=south] {\scriptsize$16$};
    \filldraw[fill=black!40,draw=black!80] (1,4) circle (1pt)    node[anchor=south] {\scriptsize$14$};
    \filldraw[fill=black!40,draw=black!80] (2,4) circle (1pt)    node[anchor=south] {\scriptsize$12$};
    \filldraw[fill=black!40,draw=black!80] (3,4) circle (1pt)    node[anchor=south] {\scriptsize$10$};
    \filldraw[fill=black!40,draw=black!80] (4,4) circle (1pt)    node[anchor=south] {\scriptsize$8$};
    \filldraw[fill=black!40,draw=black!80] (5,4) circle (1pt)    node[anchor=south] {\scriptsize$6$};
    \filldraw[fill=black!40,draw=black!80] (6,4) circle (1pt)    node[anchor=south] {\scriptsize$4$};
    \filldraw[fill=black!40,draw=black!80] (7,4) circle (1pt)    node[anchor=south] {\scriptsize$2$};
    \filldraw[fill=black!40,draw=black!80] (0,5) circle (1pt)    node[anchor=south] {\scriptsize$8$};
    \filldraw[fill=black!40,draw=black!80] (1,5) circle (1pt)    node[anchor=south] {\scriptsize$7$};
    \filldraw[fill=black!40,draw=black!80] (2,5) circle (1pt)    node[anchor=south] {\scriptsize$6$};
    \filldraw[fill=black!40,draw=black!80] (3,5) circle (1pt)    node[anchor=south] {\scriptsize$5$};
    \filldraw[fill=black!40,draw=black!80] (4,5) circle (1pt)    node[anchor=south] {\scriptsize$4$};
    \filldraw[fill=black!40,draw=black!80] (5,5) circle (1pt)    node[anchor=south] {\scriptsize$3$};
    \filldraw[fill=black!40,draw=black!80] (6,5) circle (1pt)    node[anchor=south] {\scriptsize$2$};
    \filldraw[fill=black!40,draw=black!80] (7,5) circle (1pt)    node[anchor=south] {\scriptsize$1$};

    \node [below] at (0,0) {\scriptsize$0$};
    \node [below] at (1,0) {\scriptsize$1$};
    \node [below] at (2,0) {\scriptsize$2$};
    \node [below] at (3,0) {\scriptsize$3$};
    \node [below] at (4,0) {\scriptsize$4$};
    \node [below] at (5,0) {\scriptsize$5$};
    \node [below] at (6,0) {\scriptsize$6$};
    \node [below] at (7,0) {\scriptsize$7$};
    \node [left] at (0,0) {\scriptsize$0$};
    \node [left] at (0,1) {\scriptsize$1$};
    \node [left] at (0,2) {\scriptsize$2$};
    \node [left] at (0,3) {\scriptsize$3$};
    \node [left] at (0,4) {\scriptsize$4$};
    \node [left] at (0,5) {\scriptsize$5$};

    \end{tikzpicture}
    \caption{$\Delta_5$}
    \end{subfigure}
    
    \end{multicols}
    \caption{Sets $\Delta_3$, $\Delta_4$ and $\Delta_5$, where $m=2$, $a_1=8$ and $a_2=6$.}
    \label{fig:Deltat}
\end{figure}
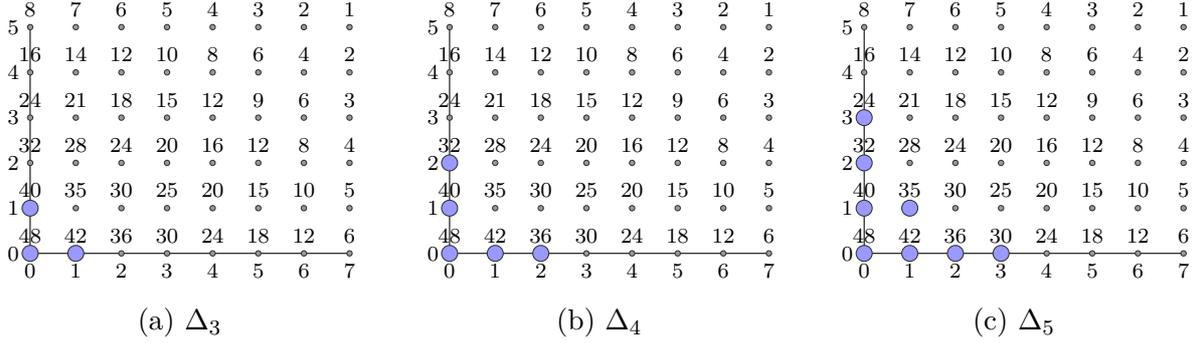

\begin{lemma}\label{le:distaceDual}
Let $\Delta_t\subseteq E$ be the set introduced in \Cref{Deltat}. Then,
$$\dis\left((C_{\boldsymbol{1},\Delta_t})^{\perp_e}\right)\geq t.$$
\end{lemma}
\begin{proof}
Using the notations in \cite[Section 3]{GGHR} the authors define a code $C(L_2)$, where
$$L_2=\{X_1^{i_1}\cdots X_m^{i_m}\in \Delta(s_1,\dots,s_m) \mid D^\perp(X_1^{i_1}\cdots X_m^{i_m})<\delta^\perp\}.$$
By choosing their $(s_1,\dots,s_m)$ and $\delta^\perp$ equal to our $(a_1,\dots,a_m)$ and $t$, respectively, then we have that
$$L_2=\{X^{\boldsymbol{e}} \mid \boldsymbol{e}\in\Delta_t\},$$
so $C(L_2)=C_{\boldsymbol{1},\Delta_t}$, see \cite[Definition 15]{GGHR}. The statement follows from their equation (8) in Section 3.
\end{proof}

\begin{theorem}\label{main2}
    Let $q$ be an odd prime power and let $m\geq 1$, $\lambda \mid q-1$, $a_1:=\lambda(q+1)$ and $2\leq a_j \leq q^2-1$, $j=2,\dots,m$ be positive integers.
    Let $n:=a_1\cdots a_m$. Consider the twist vector $\boldsymbol{v}$ defined in \eqref{twist1}, a positive integer
    $$2\leq t \leq \frac{q+3}{2}$$
    and the set $\Delta_t\subseteq E$ introduced in \Cref{Deltat}. Then, the following inclusion holds
    $$C_{\boldsymbol{v},\Delta_t}\subseteq (C_{\boldsymbol{v},\Delta_t})^{\perp_h}.$$ Therefore, there exists a stabilizer quantum code with parameters
    $$[[n,n-2\#\Delta_t, \geq t]]_q.$$
\end{theorem}

\begin{proof}
    
Let $\boldsymbol{e}\in \Delta_t$. From $\prod_{j=1}^m (e_j+1)<t$ we have that $e_1< t-1$. Since $t \le \frac{q+3}{2}$, then $e_1< t-1\le \frac{q+1}{2}$ and therefore $\Delta_t\subseteq E_0$. So, from \Cref{te:SelfOrthogonal2} we have that $C_{\boldsymbol{v},\Delta_t}\subseteq (C_{\boldsymbol{v},\Delta_t})^{\perp_h}$. 

    The existence and parameters of the stabilizer quantum code follows from \Cref{th: quantherm}. Notice that from \Cref{re:TrickForDistance} and \Cref{le:distaceDual}, we have $\dis((C_{\boldsymbol{v},\Delta_t})^{\perp_h})=\dis((C_{\boldsymbol{1},\Delta_t})^{\perp_e})\geq t$.
\end{proof}

\subsection{The Dimension}\label{dimsect}

We state a recursive formula for the dimension of the quantum code, which was shown in \cite{HoholdtHyperbolicCodes}.

Let $a, b\in\N$. Consider the case when $a_j=b$ for all $j=1, \dots, m$. We define 
\[
    V_b(m,a)\coloneqq \# \left\{(l_1, \dots , l_m) \ \middle\vert \ l_j \in \N, \  1\leq l_j \leq b, \  j = 1, \dots m, \ \prod_{j=1}^m l_j \leq a\right\}.
\]
In \cite{HoholdtHyperbolicCodes} they give the following recursive formula:
\[
    V_b(m,a)=\sum_{s=1}^{b}V\left(m-1, \floor*{\frac{a}{s}}\right),
\]
where $V_b(1,a)= \min\{a,b\}$.

    Observe that $\#\Delta_t = V_{\lambda(q+1)}(m,t-1)$, 
    where  all of $a_1, \ldots ,a_m$ are equal to $\lambda (q+1)$.
Therefore we can use the recursive formula described above to compute $\#\Delta_t$, and hence the dimension of
    the quantum code in \Cref{main2}.
For example, when $m=2$
\begin{equation}\label{delta2formula}
\#\Delta_t=V_{\lambda (q+1)}(2,t-1)=t-1+\floor*{\frac{t-1}{2}}+\floor*{\frac{t-1}{3}}+\cdots+\floor*{\frac{t-1}{t-2}}+
\floor*{\frac{t-1}{t-1}},
\end{equation}
and when $m=3$
$$\#\Delta_t=V_{\lambda (q+1)}(3,t-1)=\sum_{\alpha=1}^{t-1}
\sum_{\beta=1}^{\floor*{\frac{t-1}{\alpha}}}\floor*{\frac{t-1}{\alpha\beta}}.$$

\section{We obtain MDS and Hermitian Almost MDS quantum codes}\label{secMDS}

In this section we prove that we can obtain quantum codes that are close to the Singleton bound. Let us recall first the quantum Singleton bound.

  \begin{lemma}[Quantum Singleton bound \cite{R1999}]
   If a stabilizer quantum code with parameters $[[n,k,d]]_q$ exists, then 
   $n\geq k+2d-2$.
\end{lemma}
Codes  attaining equality are called quantum MDS codes. 

\subsection{MDS}

\begin{theorem}\label{MDS}
    The stabilizer quantum codes obtained from \Cref{main2} with $m=1$ are quantum MDS codes.
\end{theorem}

\begin{proof}
For any given bound for the minimum distance $t\in\{2,\dots, \frac{q+3}{2}\}$, we have $\Delta_t=\{0,1,2,\ldots ,t-2\}$.
The parameters of the stabilizer quantum code constructed from \Cref{main2} are
\begin{equation*}
[[n,k,d]]_q=[[\lambda (q+1),\lambda(q+1) - 2(t-1),\geq t]]_q.
\end{equation*}
It is easily verified that the above parameters provide a quantum MDS code, because $k+2d\geq \lambda(q+1) - 2(t-1)+2t=\lambda(q+1) +2=n+2$ and the quantum Singleton bound gives an equality.
\end{proof}

Some sample parameters are given in \Cref{t1,t2,t3,t4,t5}. For example, we obtain quantum MDS codes with parameters 
$[[12,8,3]]_5$ in \Cref{t2}, $[[8,4,3]]_7$ and $[[16,8,5]]_7$ in \Cref{t3} and $[[20,12,5]]_9$ in \Cref{t4}. 
We do not claim that these examples are new.

The article \cite{WanZhu} recently appeared on the arxiv and has a construction of MDS codes with lengths of the form $r(q^2-1)/h$ where $h$ is an even divisor of $q-1$ and $r\le h/2$ (their Theorems 3, 4 and 5). This article does not provide an explicit twist vector (they prove the existence of it). Our construction has an explicit twist vector and (in the $m=1$ case) gives codes with the same parameters.

\subsection{Hermitian Almost MDS}

The quantum Singleton defect of a parameter set $n,k,d$ is defined to be $n-(k+2d-2)$. 
MDS codes have quantum Singleton defect 0.
Codes with quantum Singleton defect 1 are called quantum almost MDS codes. However, from the statement of \Cref{th: quantherm} one can see that
the quantum Singleton defect of any code constructed using 
\Cref{th: quantherm} must be even, and thus a quantum Singleton defect of
1 cannot be achieved. The smallest nonzero Singleton defect 
of a code constructed using \Cref{th: quantherm} is therefore 2. This motivates the following definition.

\begin{definition}
 A quantum code constructed from \Cref{th: quantherm} with parameters $[[n,k,d]]_q$ such that $n=k+2d$ is called a
 \textit{quantum Hermitian almost MDS} (\textit{QHAMDS}) code.
\end{definition} 

In \Cref{MDS} we showed that we can construct quantum MDS codes. Recall that the quantum MDS conjecture \cite{Ketkar} states that $n\le q^2+1$ for a quantum MDS code with parameters $[[n,k,d]]_q$ and $q$ odd.  Now we are going to show that we can also construct quantum codes with $n>q^2+1$ that are at least QHAMDS. That is, they are either QHAMDS or MDS. If the quantum MDS conjecture is true, they cannot be MDS, and therefore they would have the best possible parameters.

\begin{theorem}\label{QHAMDSd3}
    The stabilizer quantum codes obtained from \Cref{main2} with $m=2$, $n>q^2+1$ and $t=3$ are at least QHAMDS.
\end{theorem}

\begin{proof}
Let $m=2$, $t=3$ and $\lambda$ and $a_2$ be as defined in \Cref{main2} such that $n>q^2+1$. We have $\Delta_3=\{(0,0),(1,0),(0,1)\}$ (see  \Cref{fig:Deltat}). The parameters of the stabilizer quantum code constructed from \Cref{main2} are
\begin{equation*}
[[n,k,d]]_q=[[\lambda (q+1)a_2,\lambda(q+1)a_2 - 6,\geq 3]]_q.
\end{equation*}
It is easily verified that the above parameters provide a code with is at least QHAMDS. This is
because $k+2d\geq\lambda(q+1)a_2 - 6+2\cdot 3=\lambda(q+1)a_2 =n$.
\end{proof}

Some examples will be given in \Cref{t1,t2,t3,t4,t5}.
In \cite{ternary} the authors study ternary quantum codes of minimum distance three. In that paper (their Theorem 4.4)  quantum codes with parameters $[[n,n-7,3]]_3$ are shown for certain lengths $n$. For those lengths which are a multiple of $4$ and less than $64$ we can  improve the dimension by 1, using the codes in \Cref{QHAMDSd3}.
See also \Cref{t1}.

\section{When $m=2$ we can beat Gilbert-Varshamov Bound}\label{secGV}

In this section we include a proof that an infinite family of codes 
obtained from our constructions will beat the quantum Gilbert-Varshamov bound when $m=2$.
We remark that the codes with $m>2$ can also beat the Gilbert-Varshamov bound, some
examples when $m=3$ are presented in \Cref{t1,t2,t4}.

Let us recall the  quantum Gilbert-Varshamov bound whose proof can be found in \cite{QuantumBounds}:

\begin{theorem}[Quantum Gilbert-Varshamov Bound]
    Suppose that $n>k \geq 2$, $d\geq 2$, and $n\equiv k \mod 2$. If
    \begin{equation}\label{qgvb}
    \frac{q^{n-k+2}-1}{q^2-1}\geq \sum_{i=1}^{d-1} (q^2-1)^{i-1} {n \choose i}
    \end{equation}
  then  there exists a pure stabilizer quantum code with parameters $[[n,k,d]]_q$.
\end{theorem}

We say that a parameter set $n,k,d, q$ beats the QGV bound if the inequality 
\eqref{qgvb} is not satisfied.

In the $m=2$ case we have the following statement, using the codes
constructed in this paper.
In this statement we are using the formula \eqref{delta2formula}.

\begin{theorem}\label{beatGV2}
Given an odd prime power $q$, and given $d$ in the range $5\le d \le (q+3)/2$, let $n$  be in the interval
\[
\biggl((d-1)^{d-1}  \frac{q^2}{(q^2-1)^{d-1}}
q^{2(d-1) (0.7+\ln (d-1))}\biggr)^{\frac{1}{d-1}} \le n \le (q^2-1)^2
\]
and have the form $\lambda (q+1)a_2$ where
$\lambda \mid (q-1)$ and $2\leq a_2 \leq q^2-1$.
Then there exists a quantum code with parameters 
$$[[n,n-2\sum_{j=1}^{d-1}   \floor*{\frac{d-1}{j}},
\ge d]]_q$$ and this code beats the quantum Gilbert-Varshamov bound.
\end{theorem}

\begin{proof}
We use the codes whose existence is proved in \Cref{main2} in the case $m=2$.
The upper bounds $d\le (q+3)/2$ and  $n \le (q^2-1)^2$ follow from the construction
in \Cref{main2}.

Let 
$$
A=\sum_{i=1}^{d-1} (q^2-1)^{i-1} {n \choose i}
$$
and let
$$ D= \frac{q^{n-k+2}-1}{q^2-1}
$$
where $k=n-2\sum_{j=1}^{d-1}   \floor*{\frac{d-1}{j}}$ 
(this comes from \eqref{delta2formula} which uses our construction with $m=2$).
We wish to prove that $A>D$ under the stated hypotheses. 
To prove this, we are going to let
$$
B=\frac{1}{(d-1)^{d-1}} n^{d-1} (q^2-1)^{d-2}
$$
and let
$$
C=\biggl( \frac{q^2}{q^2-1} \biggr) \ q^{2(d-1) (0.7 + \ln (d-1))}
$$
and we will prove three things:
that $A>B$, that $B\ge C$, and that $C>D$. 
This will complete the proof that $A>D$. 

To show that $A>B$, we will use the estimate for binomial
coefficients ${n \choose k} > (\frac{n}{k})^k$. Then
\begin{align*}
    A=\sum_{i=1}^{d-1} (q^2-1)^{i-1} {n \choose i} &> {n \choose d-1} (q^2-1)^{d-2}\\
    &> \biggl( \frac{n}{d-1}\biggr)^{d-1} (q^2-1)^{d-2}\\
    &=\frac{1}{(d-1)^{d-1}} n^{d-1} (q^2-1)^{d-2}=B.
\end{align*}

To prove that $B\ge C$, rearranging the hypothesis
\[
\biggl((d-1)^{d-1}  \frac{q^2}{(q^2-1)^{d-1}}
q^{2(d-1) (0.7+\ln (d-1))}\biggr)^{\frac{1}{d-1}} \le n
\]
yields precisely that $B\ge C$.

To prove that $C>D$, we will use the fact that if $r\ge 4$ then
$H_r < 0.7+\ln r$ where $H_r$ is the $r$-th harmonic number defined by
$H_r = \sum_{j=1}^r \frac{1}{j}$. Then
\begin{align*}
\sum_{j=1}^{d-1}   \floor*{\frac{d-1}{j}} &<  \sum_{j=1}^{d-1} \frac{d-1}{j} \\
&=(d-1) H_{d-1}\\
&< (d-1) (0.7 + \ln (d-1)) \qquad \text{since $d-1\ge 4$}.
\end{align*}
It follows that 
\begin{align*}
    D&=\frac{q^{n-k+2}-1}{q^2-1} \\ 
    &< \frac{q^{n-k+2}}{q^2-1} \\
    &= \biggl( \frac{q^2}{q^2-1} \biggr) \ q^{n-k}\\
    &= \biggl( \frac{q^2}{q^2-1} \biggr) \ q^{2\sum_{j=1}^{d-1}   \floor*{\frac{d-1}{j}}}\\
    &< \biggl( \frac{q^2}{q^2-1} \biggr) \ q^{2(d-1) (0.7 + \ln (d-1))}=C.\qedhere
\end{align*}
\end{proof}

In this theorem we assumed that $d\ge 5$ because of the constant $0.7$, which 
is a  choice. The cases $d=3$ and $d=4$ can be proved separately.
They could be included in the proof above but the constant $0.7$ would have
to be larger.
Similarly, we could have stated the theorem for $d\ge 6$ and the constant would be
smaller, it would be $0.68$. Then the $d=5$ case would need to be handled separately.
As $d$ gets larger, the constant gets smaller and approaches the 
Euler-Mascheroni constant.

\bigskip
We show \Cref{ranges} where for each $q$ between 7 and 17 and $d=5,6,7$ we give the
range of values of $n$ for which the quantum Gilbert-Varshamov bound is beaten, as given by 
\Cref{beatGV2}.

    \begin{table}[H]
\centering
\begin{tabular}{|c|| c|c |c|c|c|} 
 \hline
 \diagbox{$d$}{$q$}   & 7 & 9 & 11 &13&17\\ 
 \hline
 \hline
 5   & 742-2304  & 1438-6400 &  2450-14400& 3818-28224 & 7800-82944\\
  \hline
  6  & $d>\frac{q+3}{2}$  & 3848-6400 &  7022-14400 & 11600-28224 &26006-82944\\
  \hline
  7   & $d>\frac{q+3}{2}$  & $d>\frac{q+3}{2}$ &  None &  None & 72590-82944\\
  \hline
\end{tabular}
\caption{\centering Some instances of the range of lengths of codes (from \Cref{beatGV2} only) that beat the quantum Gilbert-Varshamov bound.}
\label{ranges}
\end{table}

A separate special analysis for each $d$, or using better
estimates in the proof, or using a computer, will give a  better range of values for $n$
than the statement of \Cref{beatGV2}.
For example, when $q=7$ and $d=5$, computer calculations show that 
the Gilbert-Varshamov bound is beaten by our codes as soon as $n>295$,
whereas the proof of \Cref{beatGV2} gives $n\ge 742$.
As another example,  when $q=11$ and $d=7$,  
the range of values of $n$ as given by the statement of \Cref{beatGV2}  is empty
(in the table we wrote `none').
However, there are in fact values of $n$
that beat the Gilbert-Varshamov bound. We state one example 
$[[7200,7172,7]]_{11}$ in \Cref{t5}.

We also remark that \Cref{beatGV2} is for $m=2$. 
A similar result will hold for $m>2$.

\subsection{$d=3$}

In the previous theorem we assumed that $d\ge 5$ to obtain a slightly stronger 
statement.
We will treat the case that  $d=3$ (and $m=2$) separately, and we will complete the 
analysis in detail now. We omit the $d=4$ case, which is similar.

Suppose $d=3$. 
By the formula \eqref{delta2formula} we have that $\Delta_3$ has 3 elements,
see also \Cref{fig:Deltat}.
The two sides of the Gilbert-Varshamov bound become
$$  \frac{q^{n-k+2}-1}{q^2-1}=  \frac{q^{8}-1}{q^2-1}
=q^6+q^4+q^2+1$$
and
$$
\sum_{i=1}^{d-1} (q^2-1)^{i-1} {n \choose i}=
n+{n \choose 2}(q^2-1).
$$
To beat the G-V bound we obtain a condition which is a
 quadratic polynomial in $n$, namely
 we require that
 \[
 n+{n \choose 2}(q^2-1) - (q^6+q^4+q^2+1)>0.
 \]
Solving the quadratic yields that the G-V bound is 
beaten when
\[
n>\frac{q^2-3+\sqrt{8q^8+q^4-6q^2+1}}{2(q^2-1)}.
\]
%It is easy to show that this is equivalent to 
%\[ n>\lceil \sqrt{2} q^2 \rceil +2. \]
For $m=2$ the largest possible $n$ is 
$(q-1)(q+1)(q^2-1)$.
Therefore, for each valid $n$ which is a multiple of $q+1$ between 
 $\frac{q^2-3+\sqrt{8q^8+q^4-6q^2+1}}{2(q^2-1)}$
and $(q^2-1)^2$
we obtain a code of that length  that beats the G-V bound.

We show \Cref{ranges2} where for each $q$ and $d=3$ we state the range of
values of  $n$ for which Gilbert-Varshamov bound is beaten.

    \begin{table}[H]
\centering
\begin{tabular}{|c|| c| c| c |c |c|} 
 \hline
 $q$ & 3 & 5 & 7 & 9 & 11 \\ 
 \hline
 Range of lengths & 15-64 & 38-576 & 72-2304 & 117-6400 & 174-14400\\
 \hline
 %4 & 20  & 49 & 92 & 151 & 223 \\
 %\hline
 %5 & 41 & 131 & 295  & 547 &  898\\
  %\hline
\end{tabular}
\caption{\centering Some instances of the range of lengths of codes from \Cref{main2} with $d=3$ that beat the quantum Gilbert-Varshamov bound.}
\label{ranges2}
%\end{center}
\end{table}

In the $d=4$ case (details omitted) the polynomial in $n$ would be cubic instead of quadratic.

\section{Examples}\label{secex}

\Cref{t1,t2,t3,t4,t5} show some samples of small values of the parameters of the quantum codes constructed with \Cref{main2}. For their minimum distance, we give the lower bound $t$ provided by \Cref{main2}. Recall that $q$ is an odd prime power, $a_1$ can be any $\lambda (q+1)$ where $\lambda$ is a divisor of $\frac{q-1}{2}$, and $a_2$ and $a_3$ can take any value between 2 and $q^2-1$. Note that for codes $[[n,k,d]]_q=[[n,k,\geq t]]_q$ constructed from \Cref{main2} we have $t\le \frac{q+3}{2}=3$ when $q=3$, and $t\le \frac{q+3}{2}=4$ when $q=5$. Recall also those with $n+2=k+2d$ are called MDS codes and codes with $n=k+2d$ are called QHAMDS codes. We also say in the sixth column if that code beats the Gilbert-Varshamov bound in the sense explained before \Cref{beatGV2}.

The article \cite{WanZhu} recently appeared on the arxiv and has a construction of MDS codes with lengths of the form $r(q^2-1)/h$ where $h$ is an even divisor of $q-1$ and $r\le h/2$ (their Theorems 3, 4 and 5).  Some of the MDS codes appearing in our tables may also be obtained with their construction.

\bigskip

\begin{table}[H]
\centering
\begin{tabular}{ | *{6}{c|} >{\centering\arraybackslash} p{6.6cm}|}
 \hline
  $m$ &$a_1$&$a_2$ &$a_3$& Quantum Code  & Beats QGV   & Comment\\
 \hline
 \hline
   1 &4&& & $[[4,0,3]]_3$ &      No    &    MDS    \\ \hline
    1 &8&&&   $[[8,4,3]]_3$      & Yes      & MDS \\ \hline
  %  2 &4&3&  &$[[12,6,3]]_3$    &No &  QHAMDS   \\ \hline
  % 2 &4&4&  &$[[16,10,3]]_3$    &Yes &   QHAMDS   \\ \hline
  % 2 &4&4&   &$[[16,6,4]]_3$    &Yes &     \\ \hline
  2&4&5&   &$[[20,14,3]]_3$        &Yes & QHAMDS      \\ \hline
 % 2&4&5&    &$[[20,10,4]]_3$        & &      \\ \hline
  2&4&6&    &$[[24,18,3]]_3$        & Yes&  QHAMDS     \\ \hline
 % 2&4&6&    &$[[24,14,4]]_3$        & &      \\ \hline
  2&4&7&    &$[[28,22,3]]_3$        & Yes& QHAMDS      \\ \hline
 % 2&4&7&    &$[[28,18,4]]_3$        & &      \\ \hline
 % 2&4&7&   &$[[28,12,5]]_3$        & &      \\ \hline
  2&4&8& &   $[[32,26,3]]_3$  &     Yes    &QHAMDS, beats $[[32,25,3]]_5$ in \cite{ternary}  \\ \hline
 % 2&4&8&   &  $[[32,22,4]]_3$  &         & \\ \hline
 % 2&4&8&   &  $[[32,16,5]]_3$  &         & \\ \hline
  2& 8&5&   &  $[[40,34,3]]_3$      &  Yes  &   QHAMDS, beats $[[33,23,3]]_5$ in \cite{ternary}     \\ \hline
  2& 8&6&   &  $[[48,42,3]]_3$      &  Yes  &   QHAMDS,   beats $[[48,41,3]]_5$ in \cite{ternary}   \\ \hline
  2& 8&7&    &  $[[56,50,3]]_3$      &  Yes  &   QHAMDS, beats $[[56,49,3]]_5$ in \cite{ternary}      \\ \hline
 2& 8&8&    &  $[[64,58,3]]_3$      &  Yes  &    QHAMDS , beats $[[64,57,3]]_5$ in \cite{ternary}    \\ \hline
  3 &8& 3&3  &  $[[72,64,3]]_3$       & Yes     &Beats $[[72,62,3]]_3$ in  \cite{KongConstacyclic2023} \\ \hline
 % 3 &4& 8&3  &  $[[96,88,3]]_3$       &   Yes   & \\ \hline
     3 &4& 8&4  &  $[[128,120,3]]_3$       & Yes     &Length not obtained with $m=1, 2$  \\ \hline
  %   3 &4& 8&4  &  $[[128,114,4]]_3$       &      & \\ \hline
   %  3 &4& 8&4  &  $[[128,102,5]]_3$       &      & \\ \hline
   %   3 &4& 8& 8 &  $[[256,248,3]]_3$       &   Yes  &length not obtained with $m=1, 2$  \\ \hline
        \end{tabular}
\caption{A $q=3$ sample of codes.}
\label{t1}
\end{table}
\bigskip

\bigskip

\begin{table}[H]
\centering
\begin{tabular}{ | *6{c|} >{\centering\arraybackslash} p{6.4cm}|  }
 \hline
  $m$ &$a_1$&$a_2$ &$a_3$& Quantum Code  & Beats QGV   & Comment\\
 \hline
 \hline
   1 &6&& &  $[[6,2,3]]_5$ &      No    &    MDS    \\ \hline
   1 &12&& &  $[[12,8,3]]_5$ &      Yes    &    MDS    \\ \hline
     1 &12&& &  $[[12,6,4]]_5$ &      Yes    &    MDS    \\ \hline
   %2 &6&3&   &$[[18,12,3]]_5$    &No & QHAMDS    \\ \hline
   %2 &6&4&   &$[[24,18,3]]_5$    & No& QHAMDS    \\ \hline
   %2 &12&2&   &$[[24,18,3]]_5$    & No& QHAMDS    \\ \hline
   %2&6&4&    &$[[24,14,4]]_5$        &No &      \\ \hline
  2&6&5&    &$[[30,24,3]]_5$        & No&  QHAMDS, beats $[[33,13,3]]_5$ in  \cite{Constacyclic2022} \\ \hline
 % 2&4&5& &3   &$[[20,10,4]]_5$        & &      \\ \hline
  2&6&6&    &$[[36,30,3]]_5$        & No&  QHAMDS    \\ \hline
  2&6&6&    &$[[36,26,4]]_5$        & No&  Length    not obtained with $m=1$\\ \hline
  2&6&7&    &$[[42,36,3]]_5$        & Yes& QHAMDS     \\ \hline
  2&6&13&    &$[[78,72,3]]_5$        & Yes& QHAMDS, beats $[[80,68,3]]_5$ in \cite{Constacyclic2022} \\ \hline
  2&6&13&    &$[[78,68,4]]_5$        & Yes&  Beats $[[78,60,4]]_5$ in \cite{MatrixProductCodes2018} \\ \hline
  2&6&16&    &$[[96,86,4]]_5$        & Yes&  Same as in  \cite{MatrixProductCodes2018}   \\ \hline
   2 &6& 19&  &  $[[114,104,4]]_5$       &  Yes   &  Length not obtained with $m=1$ \\ \hline
  2 &6& 22&  &  $[[132,122,4]]_5$       &  Yes   &Beats $[[132,118,4]]_5$ in \cite{HermitianDualConstacyclic2021} \\ \hline
 2&12&24&    &$[[288,282,3]]_5$        &Yes &  QHAMDS    \\ \hline
  2&12&24&    &$[[288,278,4]]_5$        &Yes &    Beats $[[288,275,4]]_5$ in  \cite{GHR2015MPC} \\ \hline
    3&24&13& 2   &$[[624,612,4]]_5$        &Yes &   Same as in  \cite{GHR2015MPC} \\ \hline
 %3&6&24&2    &$[[288,280,3]]_5$        &Yes &     \\ \hline
 3&24&24&2    &$[[1152,1144,3]]_5$        & Yes&  Length not obtained with $m=1, 2$   \\ \hline
 
        \end{tabular}

      \caption{A $q=5$ sample of codes.}
      \label{t2}
\end{table}  

\bigskip

\begin{table}[H]
\centering
\begin{tabular}{ | *6{c|} >{\centering\arraybackslash} p{6.7cm}|  }
 \hline
  $m$ &$a_1$&$a_2$ &$a_3$& Quantum Code  & Beats QGV   & Comment\\
 \hline
 \hline
   1 &8&&  & $[[8,4,3]]_7$ &      No    &    MDS    \\ \hline
   1 &16&&   &$[[16,12,3]]_7$    &Yes &   MDS    \\ \hline
   1 &16&&   &$[[16,10,4]]_7$    &Yes &   MDS    \\ \hline
      1 &16&&   &$[[16,8,5]]_7$    &Yes &   MDS    \\ \hline
   1 &24&&   &$[[24,20,3]]_7$    & Yes&   MDS, same as \cite{WanZhu}    \\ \hline
 %  2 &8&4&   &$[[32,26,3]]_7$    & No& QHAMDS    \\ \hline
  % 2 &8&4&   &$[[32,22,4]]_7$    & No&     \\ \hline
   %2 &8&4&   &$[[32,16,5]]_7$    & No&     \\ \hline
 %  2 &8&5&   &$[[40,34,3]]_7$    & No&   QHAMDS   \\ \hline
   1&48&&    &$[[48,44,3]]_7$        &Yes &     MDS   \\ \hline
  2&8&7&    &$[[56,50,3]]_7$        & No&   QHAMDS    \\ \hline
  2&8&8&    &$[[64,58,3]]_7$        & No&   QHAMDS, beats $[[65,53,3]]_7$ in \cite{Kolotolu2019QUANTUMCW}    \\ \hline
  2&8&8&    &$[[64,54,4]]_7$        & No&   Length not obtained with $m=1$  \\ \hline
  2&8&8&    &$[[64,48,5]]_7$        &No &   Beats $[[65,41,5]]_7$ in \cite{Kolotolu2019QUANTUMCW}  \\ \hline
2&8&9&    &$[[72,66,3]]_7$  & Yes&  QHAMDS, beats $[[75,63,3]]_7$ in \cite{MatrixProductCodes2018}\\ \hline
2&8&9&    &$[[72,56,5]]_7$        & No&    Beats $[[75,51,5]]_7$ in \cite{MatrixProductCodes2018}     \\ \hline
2&8&15&    &$[[120,114,3]]_7$  & Yes& QHAMDS, beats $[[126,114,3]]_7$ in \cite{Constacyclic2022} \\ \hline
2&8&21&    &$[[168,162,3]]_7$  & Yes& QHAMDS, beats $[[168,158,3]]_7$ in \cite{Constacyclic2022}  \\ \hline
2&8&21&    &$[[168,158,4]]_7$  & Yes&  Beats $[[168,152,4]]_7$ in \cite{Constacyclic2022}  \\ \hline
  2&8&25&    &$[[200,190,4]]_7$        & Yes&    Same as in \cite{MatrixProductCodes2018}   \\ \hline
 2&8&48&    &$[[384,378,3]]_7$        & Yes&    QHAMDS,  same as in \cite{CaoCui}   \\ \hline
 2&8&48&    &$[[384,374,4]]_7$        &Yes &   Same as in \cite{CaoCui}    \\ \hline
 2&8&48&    &$[[384,368,5]]_7$        & Yes &    Same as in \cite{CaoCui}   \\ \hline
  2&16&27&    &$[[432,422,4]]_7$        & Yes &   Beats $[[432,419,4]]_7$ in  \cite{GHR2015MPC}     \\ \hline
 3&16&48&2    &$[[768,760,3]]_7$        &Yes &  Length not obtained with $m=1, 2$    \\ \hline
        \end{tabular}
\caption{A $q=7$ sample of codes.}
\label{t3}
\end{table}  

\bigskip

\begin{table}[H]
\centering
\begin{tabular}{ | *6{c|} >{\centering\arraybackslash} p{5.9cm}|  }
 %\hline
 %\multicolumn{8}{|c|}{A Sample Of Codes} \\
 \hline
  $m$ &$a_1$&$a_2$ &$a_3$& Quantum Code  & Beats QGV   & Comment\\
 \hline
 \hline
   1 &10&&  & $[[10,6,3]]_9$ &      No    &    MDS    \\ \hline
   1 &20&&  & $[[20,16,3]]_9$ &      Yes    &    MDS    \\ \hline
    1 &20&&  & $[[20,14,4]]_9$ &      Yes    &    MDS    \\ \hline
       1 &20&&  & $[[20,12,5]]_9$ &      Yes    &    MDS    \\ \hline
   1 &40&&  & $[[40,36,3]]_9$ &      Yes    &    MDS    \\ \hline
   %2 &10&3&   &$[[30,24,3]]_9$    & No&  QHAMDS   \\ \hline
   %2 &10&5&   &$[[50,44,3]]_9$    &No &  QHAMDS    \\ \hline
%   2 &10&5&   &$[[50,40,4]]_9$    &No &     \\ \hline
 %  2 &10&5&   &$[[50,34,5]]_9$    &No &     \\ \hline
  % 2 &10&5&   &$[[50,30,6]]_9$    & No&     \\ \hline
   2 &10&10&   &$[[100,80,6]]_9$    &Yes &  Length not obtained with $m=1$ \\ \hline
    2 &10&24&   &$[[240,230,4]]_9$    &Yes & Beats  $[[246,228,4]]_9$ in \cite{MatrixProductCodes2018}  \\ \hline
   2 &10&55&   &$[[550,534,5]]_9$    &Yes &  Length not obtained with $m=1$ \\ \hline
 %  2&80&30&    &$[[2400,2392,3]]_9$        &Yes &    \\ \hline
   3&80&80&2    &$[[12800,12792,3]]_9$        &Yes &   Length not obtained with $m=1,2$     \\ \hline
        \end{tabular}
        \caption{A $q=9$ sample of codes.}
        \label{t4}
\end{table}  

\bigskip

\bigskip

\begin{table}[H]
\centering
\begin{tabular}{ | *6{c|} >{\centering\arraybackslash} p{6cm}|  }
 %\hline
 %\multicolumn{8}{|c|}{A Sample Of Codes} \\
 \hline
  $m$ &$a_1$&$a_2$ &$a_3$& Quantum Code  & Beats QGV   & Comment\\
 \hline
 \hline
   1 &12&& &  $[[12,8,3]]_{11}$ &      No    &    MDS    \\ \hline
   1 &12&& &  $[[12,6,4]]_{11}$ &      Yes    &    MDS    \\ \hline
   1 &12&& &  $[[12,4,5]]_{11}$ &      Yes    &    MDS    \\ \hline
   1 &60& &   &$[[60,56,3]]_{11}$    & Yes&  MDS\\   \hline
   1 &60& &   &$[[60,54,4]]_{11}$    &Yes&  MDS \\  \hline
   1 &60& &   &$[[60,52,5]]_{11}$    & Yes&  MDS \\ \hline
   2 &12&15&   &$[[180,174,3]]_{11}$    & Yes&QHAMDS,  beats  $[[183,171,3]]_{11}$ in \cite{MatrixProductCodes2018}  \\ \hline
    2 &12&15&   &$[[180,164,5]]_{11}$    &No &  Beats  $[[183,159,5]]_{11}$ in \cite{MatrixProductCodes2018}  \\ \hline
    2 &60&120&   &$[[7200,7172,7]]_{11}$    &Yes &  Length not obtained with $m=1$\\ \hline
   %%% 3 &12&112&3 &11  &[[4032,4004,7]]    & Yes&  \\ \hline
        \end{tabular}
        \caption{A $q=11$ sample of codes.}
        \label{t5}
\end{table}

\bigskip

\section*{Acknowledgements}

This publication has emanated from research conducted with the financial support of Science Foundation Ireland under Grant number 18/CRT/6049. For the purpose of Open Access, the author has applied a CC BY public copyright licence to any Author Accepted Manuscript version arising from this submission.

The second and third authors are partially supported by Grant TED2021-130358B-I00 funded by
MCIN/AEI/10.13039/501100011033 and by the “European Union
NextGenerationEU/PRTR”, as well as by Universitat Jaume I, grants UJI-B2021-02, GACUJIMA/2023/06 and PREDOC/2020/39.

The third author would also like to acknowledge the funding received from the UCD School of Mathematics and Statistics.

\bibliographystyle{plain}
\bibliography{biblio}

\begin{thebibliography}{10}

\bibitem{Aly}
S.~A. Aly, A.~Klappenecker, and P.~K. Sarvepalli.
\newblock On quantum and classical {BCH} codes.
\newblock {\em IEEE Trans. Inf. Theory}, 53(3):1183--1188, 2007.

\bibitem{ANDERSEN200892}
H.~E. Andersen and O.~Geil.
\newblock Evaluation codes from order domain theory.
\newblock {\em Finite Fields their Appl.}, 14(1):92--123, 2008.

\bibitem{ABKL2000I}
A.~Ashikhmin, A.~Barg, E.~Knill, and S.~Litsyn.
\newblock Quantum error-detection {I}: {S}tatement of the problem.
\newblock {\em IEEE Trans. Inf. Theory}, 46:778--788, 2000.

\bibitem{ABKL2000II}
A.~Ashikhmin, A.~Barg, E.~Knill, and S.~Litsyn.
\newblock Quantum error-detection {II}: {B}ounds.
\newblock {\em IEEE Trans. Inf. Theory}, 46:789--800, 2000.

\bibitem{AK}
A.~Ashikhmin and E.~Knill.
\newblock Non-binary quantum stabilizer codes.
\newblock {\em IEEE Trans. Inf. Theory}, 47:3065--3072, 2001.

\bibitem{ALT2001}
A.~Ashikhmin, S.~Litsyn, and M.~A. Tsfasman.
\newblock Asymptotically good quantum codes.
\newblock {\em Phys. Rev. A}, 63(3):032311, 2001.

\bibitem{Ball}
S.~Ball.
\newblock Some constructions of quantum {MDS} codes.
\newblock {\em Des. Codes Cryptogr.}, 89:811--821, 2021.

\bibitem{Constacyclic2022}
S.~Bhardwaj, M.~Goyal, and M.~Raka.
\newblock New quantum codes from constacyclic codes over a general non-chain
  ring.
\newblock arXiv preprint
  \href{https://arxiv.org/abs/2212.02821}{arXiv:2212.02821}, 2022.

\bibitem{BE}
J.~Bierbrauer and Y.~Edel.
\newblock Quantum twisted codes.
\newblock {\em J. Comb. Designs}, 8:174--188, 2000.

\bibitem{Brooks23}
M.~Brooks.
\newblock Quantum computers: what are they good for?
\newblock {\em Nature}, 617:S1--S3, 2023.

\bibitem{18kkk}
A.~R. Calderbank, E.~M. Rains, P.~W. Shor, and N.~J.~A. Sloane.
\newblock Quantum error correction and orthogonal geometry.
\newblock {\em Phys. Rev. Lett.}, 76:405--409, 1997.

\bibitem{Calderbank}
A.~R. Calderbank, E.~M. Rains, P.~W. Shor, and N.~J.~A. Sloane.
\newblock Quantum error correction via codes over {${\rm GF}(4)$}.
\newblock {\em IEEE Trans. Inf. Theory}, 44(4):1369--1387, 1998.

\bibitem{CLMS2021}
E.~Camps, H.~H. López, G.~L. Matthews, and E.~Sarmiento.
\newblock Polar decreasing monomial-{C}artesian codes.
\newblock {\em IEEE Trans. Inf. Theory}, 67(6):3664--3674, 2021.

\bibitem{CaoCui}
M.~Cao and J.~Cui.
\newblock Construction of new quantum codes via {H}ermitian dual-containing
  matrix-product codes.
\newblock {\em Quantum Inf. Process.}, 19:427, 2020.

\bibitem{Cas2017}
D.~Castelvecchi.
\newblock Quantum computers ready to leap out of the lab in 2017.
\newblock {\em Nature}, 541(7635):9--10, 2017.

\bibitem{ternary}
G.~Chen and R.~Li.
\newblock Ternary self-orthogonal codes of dual distance three and ternary
  quantum codes of distance three.
\newblock {\em Des. Codes Cryptogr.}, 69:53--63, 2013.

\bibitem{CoxLittle}
D.~Cox, J.~Little, and D.~O’Shea.
\newblock {A}n {I}ntroduction to {C}omputational {A}lgebraic {G}eometry and
  {C}ommutative {A}lgebra.
\newblock In S.~Axler and K.~Ribet, editors, {\em Ideals, {V}arieties, and
  {A}lgorithms}, Undergraduate Texts in Mathematics, New York, 2007. Springer.

\bibitem{8AS}
D.~Dieks.
\newblock Communication by {EPR} devices.
\newblock {\em Phys. Rev. A}, 92:271, 1982.

\bibitem{Fang}
W.~Fang and F.W. Fu.
\newblock Some new constructions of quantum {MDS} codes.
\newblock {\em IEEE Trans. Inf. Theory}, 65:7840--7847, 2019.

\bibitem{FengRao}
G.-L. Feng and T.~R.~N. Rao.
\newblock Decoding algebraic-geometric codes up to the designed minimum
  distance.
\newblock {\em IEEE Trans. Inf. Theory}, 39(1):37--45, 1993.

\bibitem{Feng2002}
K.~Feng.
\newblock Quantum error correcting codes.
\newblock In H.~Niederreiter, editor, {\em Coding {T}heory and {C}ryptology},
  volume~1 of {\em Lecture Notes Series, Institute for Mathematical Sciences,
  National University of Singapore}, pages 91--142, 2002.

\bibitem{QuantumBounds}
K.~Feng and Z.~Ma.
\newblock A finite {G}ilbert-{V}arshamov bound for pure stabilizer quantum
  codes.
\newblock {\em IEEE Trans. Inf. Theory}, 50(12):3323--3325, 2004.

\bibitem{fit}
J.~Fitzgerald and R.~F. Lax.
\newblock Decoding {A}ffine {V}ariety {C}odes {U}sing {G}r\"{o}bner {B}asis.
\newblock {\em Des. Codes Cryptogr.}, 13:147--158, 1998.

\bibitem{GGHR2017}
C.~Galindo, O.~Geil, F.~Hernando, and D.~Ruano.
\newblock On the distance of stabilizer quantum codes from {$J$}-affine variety
  codes.
\newblock {\em Quantum Inf. Process.}, 16(111), 2017.

\bibitem{GGHR}
C.~Galindo, O.~Geil, F.~Hernando, and D.~Ruano.
\newblock Improved {C}onstructions of {N}ested {C}ode {P}airs.
\newblock {\em IEEE Trans. Inf. Theory}, 64(4):2444--2459, 2018.

\bibitem{GHMC2023}
C.~Galindo, F.~Hernando, and H.~Martín-Cruz.
\newblock Optimal $(r,\delta)$-{LRC}s from monomial-{C}artesian codes and their
  subfield-subcodes.
\newblock arXiv preprint
  \href{https://arxiv.org/abs/2205.01485}{arXiv:2205.01485}, 2023.

\bibitem{GHMR2023}
C.~Galindo, F.~Hernando, H.~Martín-Cruz, and D.~Ruano.
\newblock Stabilizer quantum codes defined by trace-depending polynomials.
\newblock {\em Finite Fields their Appl.}, 87:102138, 2023.

\bibitem{GHR2015MPC}
C.~Galindo, F.~Hernando, and D.~Ruano.
\newblock New quantum codes from evaluation and matrix-product codes.
\newblock {\em Finite Fields their Appl.}, 36:98--120, 2015.

\bibitem{OH}
O.~Geil and T.~Hoholdt.
\newblock Footprints or generalized {B}ezout's theorem.
\newblock {\em IEEE Trans. Inf. Theory}, 46(2):635--641, 2000.

\bibitem{HoholdtHyperbolicCodes}
O.~Geil and T.~H{\o}holdt.
\newblock On hyperbolic codes.
\newblock In S.~Boztas and I.~E. Shparlinski, editors, {\em Applied Algebra,
  Algebraic Algorithms and Error-Correcting Codes}, volume 2227 of {\em Lecture
  Notes in Computer Science}, pages 159--171, Berlin, Germany, 2001. Springer.

\bibitem{Gottesman}
D.~Gottesman.
\newblock Class of quantum error-correcting codes saturating the quantum
  {H}amming bound.
\newblock {\em Phys. Rev. A}, 54(3):1862--1868, 1996.

\bibitem{GBR2004}
M.~Grassl, T.~Beth, and M.~R{ö}tteler.
\newblock On optimal quantum codes.
\newblock {\em Int. J. Quantum Inf.}, 2(1):55--64, 2004.

\bibitem{45kkk}
M.~Grassl and M.~R\"{o}tteler.
\newblock Quantum {BCH} codes.
\newblock In {\em Proc. X Int. Symp. Theor. Elec. Eng.}, pages 207--212, 1999.

\bibitem{H-VL-P}
T.~H{\o}holdt, J.~H. {van Lint}, and G.~R. Pellikaan.
\newblock Algebraic geometry codes.
\newblock In V.~S. Pless and W.~C. Huffman, editors, {\em Handbook of Coding
  Theory}, volume~1, pages 871--961, Netherlands, 1998. Elsevier.

\bibitem{Ketkar}
A.~Ketkar, A.~Klappenecker, S.~Kumar, and P.~K. Sarvepalli.
\newblock Nonbinary {S}tabilizer {C}odes {O}ver {F}inite {F}ields.
\newblock {\em IEEE Trans. Inf. Theory}, 52(11):4892--4914, 2006.

\bibitem{Kolotolu2019QUANTUMCW}
E.~Kolotoğlu and M.~Sarı.
\newblock Quantum codes with improved minimum distance.
\newblock {\em Bull. Korean Math. Soc.}, 56(3):609--619, 2019.

\bibitem{KongConstacyclic2023}
B.~Kong and X.~Zheng.
\newblock Quantum codes from constacyclic codes over {$S_k$}.
\newblock {\em EPJ Quantum Technol.}, 10(3), 2023.

\bibitem{lag1}
G.~G. La~Guardia.
\newblock Construction of new families of nonbinary quantum codes.
\newblock {\em Phys. Rev. A}, 80:042331, 2009.

\bibitem{lag2}
G.~G. La~Guardia.
\newblock On the {C}onstruction of {N}onbinary {Q}uantum {BCH} {C}odes.
\newblock {\em IEEE Trans. Inf. Theory}, 60(3):1528--1535, 2014.

\bibitem{Liu-LiuX}
H.~Liu and X.~Liu.
\newblock Constructions of quantum {MDS} codes.
\newblock {\em Quantum Inf. Process.}, 20(14), 2021.

\bibitem{MatrixProductCodes2018}
X.~Liu, H.~Q. Dinh, H.~Liu, and L.~Yu.
\newblock On new quantum codes from matrix product codes.
\newblock {\em Cryptogr. Commun.}, 10:579--589, 2018.

\bibitem{LMS2020}
H.~H. López, G.~L. Matthews, and I.~Soprunov.
\newblock Monomial-{C}artesian codes and their duals, with applications to
  {LCD} codes, quantum codes, and locally recoverable codes.
\newblock {\em Des. Codes Cryptogr.}, 88:1673--1685, 2020.

\bibitem{LSV2021}
H.~H. López, I.~Soprunov, and R.~H. Villarreal.
\newblock The dual of an evaluation code.
\newblock {\em Des. Codes Cryptogr.}, 89:1367--1403, 2021.

\bibitem{71kkk}
R.~Matsumoto and T.~Uyematsu.
\newblock Constructing {Q}uantum {E}rror-{C}orrecting {C}odes for $p^m$-{S}tate
  {S}ystems from {C}lassical {E}rror-{C}orrecting {C}odes.
\newblock {\em IEICE Trans. Fundam. Electron. Commun. Comput. Sci.},
  E83-A(10):1878--1883, 2000.

\bibitem{R1999}
E.~M. Rains.
\newblock Quantum weight enumerators.
\newblock {\em IEEE Trans. Inf. Theory}, 44(4):1388--1394, 1998.

\bibitem{23RBC}
P.~W. Shor.
\newblock Scheme for reducing decoherence in quantum computer memory.
\newblock {\em Phys. Rev. A}, 52(4):2493--2496, 1995.

\bibitem{Shor2}
P.~W. Shor.
\newblock Polynomial-{T}ime {A}lgorithms for {P}rime {F}actorization and
  {D}iscrete {L}ogarithms on a {Q}uantum {C}omputer.
\newblock {\em SIAM J. Comput.}, 26(5):1484--1509, 1997.

\bibitem{Song}
H.~Song, R.~Li, Y.~Liu, and G.~Guo.
\newblock New quantum codes from matrix-product codes over small fields.
\newblock {\em Quantum Inf. Process.}, 19(226), 2020.

\bibitem{St96}
A.~Steane.
\newblock Multiple-particle interference and quantum error correction.
\newblock {\em Proc. R. Soc. Lond. A}, 452:2551--2577, 1996.

\bibitem{95kkk}
A.~M. Steane.
\newblock Simple quantum error-correcting codes.
\newblock {\em Phys. Rev. A}, 54(6):4741--4751, 1996.

\bibitem{WanZhu}
R.~Wan and S.~Zhu.
\newblock New {Q}uantum {MDS} codes from {H}ermitian self-orthogonal
  generalized {R}eed-{S}olomon codes.
\newblock arXiv preprint
  \href{https://arxiv.org/abs/2302.06169}{arXiv:2302.06169}, 2023.

\bibitem{HermitianDualConstacyclic2021}
Y.~Wang, X.~Kai, Z.~Sun, and S.~Zhu.
\newblock Quantum codes from {H}ermitian dual-containing constacyclic codes
  over ${\F}_{q^{2}}+{v}{\F}_{q^{2}}$.
\newblock {\em Quantum Inf. Process.}, 20(122), 2021.

\bibitem{26RBC}
W.~K. Wootters and W.~H. Zurek.
\newblock A single quantum cannot be cloned.
\newblock {\em Nature}, 299:802--803, 1982.

\end{thebibliography}
\end{document}